%% file: main_ieee.tex
\documentclass[journal]{IEEEtran}
\usepackage{indentfirst}
\usepackage{amsfonts}
\usepackage{amsmath}
\usepackage{amsthm}
\usepackage{graphicx}
\usepackage{fancyhdr}
\usepackage{float}
\usepackage{commath}
\usepackage[utf8]{inputenc}
\usepackage{amssymb}
\usepackage{enumerate}
\usepackage[shortlabels]{enumitem}
\usepackage{xcolor}
\usepackage{soul}
\usepackage{listings}
\usepackage[hidelinks,colorlinks=true,linkcolor=blue,citecolor=blue,urlcolor=blue]{hyperref}
\usepackage{bm}
\usepackage{cite}
\usepackage{flushend}
\usepackage{lipsum}
\usepackage{adjustbox}
\usepackage{subcaption}
\usepackage{tabularx}
\usepackage{tikz}
\usepackage{pgfplots}
\usetikzlibrary{patterns}
\usepgfplotslibrary{fillbetween,colorbrewer}
\usetikzlibrary{positioning,shadows,shapes,backgrounds,calc,spy,arrows}
\usetikzlibrary{intersections}
\usepackage{pgfplotstable}
\pgfplotsset{compat=newest}
\usepackage{tcolorbox}

\usepackage{caption}
\newtheorem{theorem}{Theorem}
\newtheorem{corollary}{Corollary}

\definecolor{codegreen}{rgb}{0,0.6,0}
\definecolor{codegray}{rgb}{0.5,0.5,0.5}
\definecolor{codepurple}{rgb}{0.58,0,0.82}
\definecolor{backcolour}{rgb}{0.99,0.99,0.99}

\lstdefinestyle{mystyle}{
    backgroundcolor=\color{backcolour},   
    commentstyle=\color{codegreen},
    keywordstyle=\color{blue},
    numberstyle=\tiny\color{codegray},
    stringstyle=\color{codepurple},
    basicstyle=\footnotesize\ttfamily,
    breakatwhitespace=false,         
    breaklines=true,                 
    captionpos=b,                    
    keepspaces=true,                 
    numbers=left,                    
    numbersep=5pt,                  
    showspaces=false,                
    showstringspaces=false,
    showtabs=false,                  
    tabsize=2
}
 
\renewcommand{\Pr}[1]{\mathrm{Pr}\left\{#1\right\}}

\newcommand{\thn}{\mathrm{ele}}
\newcommand{\rin}{\mathrm{rin}}
\newcommand\notsotiny{\@setfontsize\notsotiny\@vipt\@viipt}
\newcommand{\Pcw}{\overline{P}_{\mathrm{cw}}}

\DeclareCaptionFont{smallcaps}{\scshape}  
\begin{document}

\title{Beyond 200 Gb/s/lane: An Analytical Approach to Optimal Detection in Shaped IM-DD Optical Links with Relative Intensity Noise}

\author{Felipe~Villenas,~\IEEEmembership{Student Member,~IEEE,} Kaiquan~Wu,~\IEEEmembership{Member,~IEEE,} Yunus~Can~G\"{u}ltekin,~\IEEEmembership{Member,~IEEE,} Jamal~Riani, and Alex~Alvarado,~\IEEEmembership{Senior Member,~IEEE}
\thanks{This research is part of the project COmplexity-COnstrained LIght-coherent optical links (COCOLI) funded by Holland High Tech $|$ TKI HSTM via the PPS allowance scheme for public-private partnerships. \textit{(Corresponding author: Felipe Villenas.)}}%
\thanks{Felipe Villenas, Kaiquan Wu, Yunus Can G\"{u}ltekin, and Alex Alvarado are with the Information and Communication Theory Lab, Signal Processing Systems Group, Department of Electrical Engineering, Eindhoven University of Technology, 5600 MB Eindhoven, The Netherlands (e-mails: \{f.i.villenas.cortez, k.wu, y.c.g.gultekin, a.alvarado\}@tue.nl).}%
\thanks{Jamal Riani is with Marvell Technology, Santa Clara, CA 95054-3606 USA (e-mail: rjamal@marvell.com).}%
}



\maketitle


\begin{abstract}

Next-generation intensity-modulation (IM) and direct-detection (DD) systems used in data centers are expected to operate at 400 Gb/s/lane and beyond. Such rates can be achieved by increasing the system bandwidth or the modulation format, which in turn requires maintaining or increasing the signal-to-noise ratio (SNR). Such SNR requirements can be achieved by increasing the transmitted optical power. This increase in optical power causes the emergence of relative intensity noise (RIN), a signal-dependent impairment inherent to the transmitter laser, which ultimately limits the performance of the system. In this paper, we develop an analytical symbol error rate (SER) expression for the optimal detector for the IM-DD optical link under study. The developed expression takes into account the signal-dependent nature of RIN and does not make any assumptions on the geometry or probability distribution of the constellation. Our expression is therefore applicable to general probabilistically and/or geometrically shaped systems. Unlike results available in the literature, our proposed expression provides a perfect match to numerical simulations of probabilistic and geometrically shaped systems.

\end{abstract}


\IEEEpeerreviewmaketitle


\input{sections/1.introduction}
\input{sections/2.system_model}
\input{sections/3.symbol_error}
\input{sections/4.simulations}

\input{sections/5.conclusions}
\input{sections/A1.appendix}
\bibliographystyle{IEEEtran}
\bibliography{references}


\end{document}

%% file: sections/1.introduction.tex
\section{Introduction}


\IEEEPARstart{F}{ueled} partly by the exploding demand from artificial intelligence applications, the need for increasing data rates in short-reach data center interconnects (DCI) is unstoppable \cite{zhou2019beyond}. Maintaining hardware costs relatively low is also of importance, as the vast majority of these optical links employ cost-effective transceivers based on intensity-modulation (IM) and direct-detection (DD) \cite{che2023modulation}. With IM-DD, only the amplitude of the transmitted optical light is modulated to carry information. Pulse amplitude modulation (PAM) formats are typically employed in these systems, allowing low-complexity transceivers with low power consumption \cite{zhong2018digital}.

Current commercial IM-DD systems have achieved up to 200 Gb/s/lane \cite{800G_MSA_}. Increasing the symbol rate, which in turn requires an increase in bandwidth, is the most straightforward option for achieving higher data rates (400 Gb/s/lane and beyond). However, achieving such rates is mostly hindered by the constraints imposed by the bandwidth of the electronic components, chip packaging and PCB routing within the transceivers. An alternative approach to improve spectral efficiency that does not require a large bandwidth increase, is to increase the cardinality $M$ of the modulation format PAM-$M$ \cite{berikaa2022net,hossain2021single}. However, modulation formats beyond PAM-4 have a lower noise tolerance due to the higher number of PAM levels. In both cases (increasing bandwidth or modulation format cardinality) requires an increase in optical power to maintain or improve the signal-to-noise ratio (SNR). 

In IM-DD optical links, the increase in signal power leads to signal-dependent noise \cite{szczerba20124}. Without optical amplification, relevant signal-dependent noise sources to take into consideration are the shot noise generated at the receiver \cite{safari2015efficient,van2018optimization}, and the relative intensity noise (RIN) generated from the transmitter laser \cite{baveja201756,villenas2025ecoc}. The latter is the dominant noise source in the high optical power regime \cite[Sec.~II]{szczerba20124}. Consequently, one of the main challenges for achieving 400 Gb/s/lane and beyond, is to address signal-dependent noise. Due to the reasons above, here we focus on IM-DD systems without optical amplification subject to the presence of RIN. 

To compensate for both transceiver impairments and increase throughput in short-reach optical links, an increasing number of digital signal processing (DSP) techniques have been implemented into IM-DD systems \cite{wettlin2020dsp}. An example of a modern DSP technique that aims to increase the throughput in optical links is constellation shaping. This DSP technique has been vastly studied in long-haul coherent optical links, and has been gaining traction for IM-DD systems as well in recent years. The most common form of shaping is probabilistic shaping (PS) \cite{bocherer2015bandwidth}, where the probability distribution of the transmitted symbols is modified to achieve a higher achievable information rate or rate adaptation, for the same transmitted power \cite{fehenberger2016probabilistic}. On the other hand, geometric shaping (GS) changes the geometry of the constellation in order to also achieve higher information rates \cite{qu2019probabilistic,chen2023orthant}. In long-haul coherent optical communications, the inclusion of constellation shaping targets information rate maximization in both the linear and nonlinear operation regimes of the optical channel, see for example \cite{bocherer2019probabilistic,liga2022model} and references therein. In IM-DD systems, constellation shaping has also been studied for short-reach applications. For example, PAM formats with PS were considered in \cite{hossain2021single,ozaydin2024optimization,ozaydin2025optimal,liang2025probabilistic} for information rate maximization, and GS with symbol error minimization in \cite{katz2018level,liang2023geometric,villenas2025ofc}. 

The main metric to predict and assess the performance of an IM-DD link using hard-decision forward error correction (FEC) is the pre-FEC bit error rate (BER). In many works, the pre-FEC BER of a system is either obtained empirically or analytically approximated from the symbol error rate (SER) expressions under certain assumptions, e.g., sufficiently high SNR. The optimal method to detect the transmitted symbols at the receiver is to utilize the maximum a-posteriori probability (MAP) decision rule. For high data rates IM-DD links where the impairments caused by the signal-dependent noise sources are nonnegligible, obtaining an accurate estimation of the pre-FEC BER or SER is critical for performance assessment. Furthermore, when combined with constellation shaping, the optimal decision rule for these systems is in general unknown.

Analytical expressions for the SER of PAM-$M$ signals in the literature, e.g. \cite{safari2015efficient}, often assume equally spaced and equiprobable symbols. Moreover, the fiber channel is typically modeled as an additive white Gaussian noise (AWGN) channel. However, under these assumptions, the developed expressions do not capture the effects of signal-dependent noise and constellation shaping. A few works in the literature go beyond these assumptions. For example, \cite{anttonen2011error} compares different SER expressions for PAM signals with signal-dependent noise in multipath fading channels. One of the SER expressions considers Gaussian distribution for the channel noise. Then, the thresholds' calculation takes into account GS by utilizing the symbols' positions, but assume equally-likely symbols, thus making it incompatible with PS. Similarly, the work in \cite{chagnon2014experimental} performs an experimental study of a short-reach IM-DD system and compares their measurements with an analytical SER expression for PAM signals. Their SER expression takes into account both arbitrary symbol's position and signal-dependent noise, but also assumes equally-likely symbols. Recently, a general MAP decision rule was derived for PAM-4 signals in optical fiber channels in \cite{li2023application}. The authors use a nonequiprobable bit sequence as a test pattern, and thus, a general SER expression is developed that is compatible with PS. However, \cite{li2023application} neglects the signal-dependence of the noise which results in decision thresholds for a regular AWGN channel. The work in \cite{zhou2022unequally} also studies an optical link with PAM-4 where shot noise is the source of signal-dependent noise caused by the usage of a simplified coherent receiver. The authors explore SER minimization with GS by using an analytical expression that has an arbitrary constellation spacing, but assume equiprobable symbols. Table~\ref{tab:sota} shows a summary of available SER expressions in the literature, showing their compatibility with GS and PS. Table~\ref{tab:sota} also shows whether the considered work considers signal-dependent noise, denoted as $\sigma^2(X)$, and the modulation format studied.

In this paper, we derive a general analytical SER expression for an IM-DD system with PAM-$M$ that takes into account the effect of signal-dependent noise caused by RIN from the transmitter laser. Furthermore, we make no assumptions on the geometry nor the distribution of the constellation, and thus, our expression is applicable for GS and/or PS. To the best of our knowledge, this is the first general expression (see Table~\ref{tab:sota}), and thus, provides accurate SERs for the signal-dependent noise channel when also employing constellation shaping techniques.

The rest of the paper is organized as follows. Sec.~\ref{sec:system_model} describes the source of RIN and the IM-DD system under study, as well as the derived channel model for this specific system. Sec.~\ref{sec:symbol_error} details the derivation of the analytical SER expression for the considered signal-dependent noise channel. Sec.~\ref{sec:results} includes numerical results, where the developed expression is compared against the literature (e.g., \cite{chagnon2014experimental}) under different constellation shaping scenarios. Finally, Sec.~\ref{sec:conclusions} presents the conclusions of this work.

\newcolumntype{C}{>{\centering\arraybackslash}X}

\begin{table}[!t]
    \centering
    \caption{Comparison of SER Expressions in the Literature}
    \begin{tabularx}{\columnwidth}{>{\centering\arraybackslash}p{1.9cm}CCCC}
        \hline \hline
        Work & GS & PS & $\sigma^2(X)$ & Modulation\\
        \hline
        \cite{zhou2022unequally} & Yes & No & Yes & PAM-4\\
        \cite{anttonen2011error,chagnon2014experimental} & Yes & No & Yes & PAM-$M$\\
        \cite{li2023application} & Yes & Yes & No & PAM-$M$\\
        \textbf{This work} & \textbf{Yes} & \textbf{Yes} & \textbf{Yes} & \textbf{PAM-}$\boldsymbol{M}$\\
        \hline \hline
    \end{tabularx}
    \label{tab:sota}
\end{table}

%% file: sections/2.system_model.tex
\section{System Model}
\label{sec:system_model}

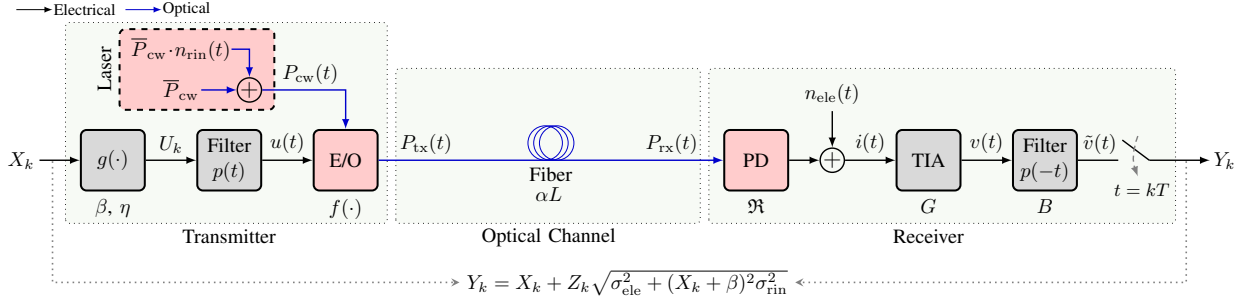
\begin{figure*}[!ht]
    \centering
    \begin{adjustbox}{width=0.9\linewidth,trim={1.95cm 0 0.25cm 0}, clip}
        \input{FigTikz/block_diagram_IMDD2}
    \end{adjustbox}
    \caption{General structure of the considered IM-DD optical fiber link. The transmitter consists of a laser modulated by an electrical-to-optical (E/O) modulator. The optical channel only considers fiber losses. The receiver consists of a photodiode (PD) followed by a transimpedance amplifier (TIA), matched filter $p(-t)$ and sampling.}
    \label{fig:imdd_link}
\end{figure*}

We consider an IM-DD link without optical amplification, as shown in Fig.~\ref{fig:imdd_link}. We also consider that all the components have an ideal frequency response and none is bandwidth-limited. Hence, no inter-symbol interference is introduced to the signals, as discussed in \cite[Sec.~IV]{che2023modulation}. In the following sections, the components of the optical system in Fig.~\ref{fig:imdd_link} are explained in more detail. 

\subsection{Transmitter} \label{sec:transmitter}

At the transmitter side, a sequence of PAM-$M$ symbols $X_k\in\mathcal{X}$ is generated randomly. To mitigate the nonlinear transfer function of the eletrical-to-optical (E/O) modulator, we utilize a pre-distortion function $g(\cdot)$ such that it generates $U_k=g(X_k)$\footnote{Notation convention: Capital letters $X$ denote a random variable (RV), while their lowercase version $x$ denotes the RV realization. $\mathbb{E}[\cdot]$ denotes expectation. Calligraphic letters denote sets, $f(t)$ with $t\in\mathbb{R}$ denotes a continuous-time function, and $f_k$ with $k\in\mathbb{Z}$ denotes a discrete-time sequence.}. The symbols $U_k$ are then passed through an ideal filter $p(t)$ to generate the analog waveform
\begin{equation} \label{eq:U_k}
    u(t) = \sum_{k=0}^{\infty}U_k p(t-kT) = \sum_{k=0}^{\infty}g(X_k)p(t-kT),
\end{equation}
where $p(t)$ is a raised cosine pulse waveform, and $T$ is the symbol period defined by the baudrate $R_s=1/T$. 

The signal $u(t)$ is used as the driving signal of the E/O modulator, whose function is to modulate the intensity $P_{\mathrm{cw}}(t)$ of an O-band laser in continuous-wave (CW) operation. The baseband output optical power from the laser can be modeled as \cite[Eq.~(2.2)]{seimetz2009high}
\begin{equation} \label{eq:Pcw}
    P_{\mathrm{cw}}(t) = \Pcw(1 + n_{\rin}(t)),
\end{equation}
where $\Pcw$ is the laser average optical power, and $n_{\rin}(t)$ is RIN with zero mean and power spectral density (PSD) $S_{\rin}(f)$. The laser output is mostly dominated by the stimulated emission of photons. However, a small part of these photons is generated by spontaneous emission, which results in optical intensity noise $n_{\rin}(t)$ due to the incoherent nature of these photons~\cite[Ch.~3.3.6]{hui2019introduction}. This noise manifests itself as random fluctuations around the average intensity of the signal, as shown in Fig.~\ref{fig:rin_example0}.

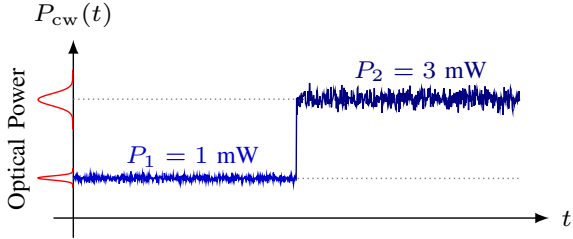
\begin{figure}[!t]
    \centering
    \begin{adjustbox}{width=0.9\columnwidth}
        \input{FigTikz/RIN_example3.tex}
    \end{adjustbox}
    \caption{Illustration of the optical power emitted by an O-band laser in CW operation with RIN. An example is given for two different power levels $P_1=1$ mW and $P_2=3$ mW.}
    \label{fig:rin_example0}
\end{figure}

This intensity noise is the RIN, and is defined as the ratio between the noise power and $\Pcw$, hence the name ``relative intensity noise''. This ratio is a commonly used parameter for laser specification, and is measured in units of [1/Hz] or [dB/Hz]. The measurement of this parameter is done in the electrical domain after detection with a photodiode (PD). Therefore, because of square-law detection in the PD, the power of the generated electrical noise signal is proportional to the square of the laser output optical power $\Pcw^2$. The laser RIN parameter is defined as \cite[Ch.~6.5.2]{agrawal2013semiconductor}
\begin{equation} \label{eq:rin}
    S_{\rin}(f) \triangleq \int_{-\infty}^{\infty}R_\rin(\tau)e^{-j2\pi f\tau}\mathrm{d}\tau
\end{equation}
where $R_\rin(\tau)=\mathbb{E}[n_\rin(t+\tau)n_\rin(t)]$ is the autocorrelation function of $n_\rin(t)$. Assuming an additive white Gaussian noise model for the RIN \cite[Ch.~3.3.6]{hui2019introduction}, $S_{\rin}(f) = \mathrm{RIN}$, i.e., a constant $\forall f$. Hence, after band-limiting the noise with a low pass filter, from \eqref{eq:Pcw} we have that the variance of the laser power is
\begin{equation} \label{eq:laser_variance}
    \begin{split}
    \mathbb{E}[(P_{\mathrm{cw}}(t)-\Pcw)^2] &= \mathbb{E}[(\Pcw\cdot n_\rin(t))^2]\\
    &=\Pcw^2\int_{0}^B \!\!S_\rin(f)\mathrm{d}f = \Pcw^2\cdot\mathrm{RIN}\cdot B
    \end{split}
\end{equation}
where $B$ is the electrical bandwidth of the band-limiting filter. In the example illustration from Fig.~\ref{fig:rin_example0} it can be observed that the variance of the intensity noise is larger for high power than it is for lower optical power. Typical laser RIN values range from $-160$ to $-130$ [dB/Hz] \cite{seimetz2009high}.

\begin{figure}[!t]
    \centering
    \input{FigTikz/noise_comparison}
    \caption{Noise contributions due to thermal noise, shot noise, and RIN for a system bandwidth of $100$ GHz, thermal noise equivalent power of \mbox{$22$ pA$/\sqrt{\textrm{Hz}}$}, photodiode responsivity \mbox{$\mathfrak{R}=0.5$ A/W}, and \mbox{$\mathrm{RIN}=-140$ dB/Hz}. The shaded area corresponds to the region of interest for high-speed IM-DD.}
    \label{fig:rin_example}
\end{figure}
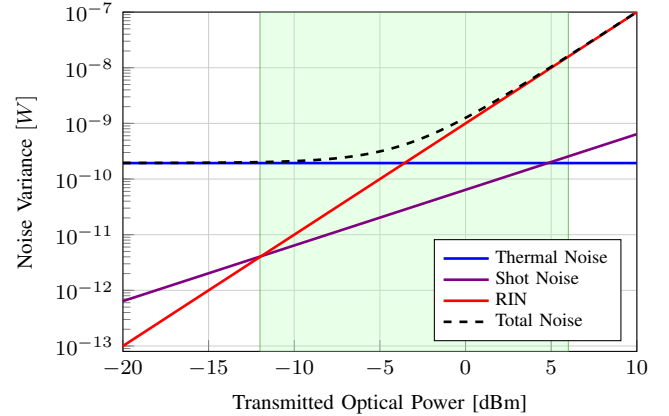

In this paper, we focus on the main source of signal-dependent noise in this IM-DD setup from Fig.~\ref{fig:imdd_link}, i.e., the transmitter laser generating intensity noise. The main noise contributions mentioned in the introduction are plotted in Fig.~\ref{fig:rin_example} as a function of the transmitted optical power. Thermal noise has a constant variance because it is independent of the transmitted optical power. In contrast, shot noise and RIN variance depend linearly and quadratically on the transmitted optical power, respectively \cite[Eq.~(10)]{szczerba20124}. Furthermore, it can be observed that the thermal noise is dominant for low optical power values, whereas for high optical power values, the RIN dominates over both thermal and shot noise. Due to this noise dominance, we focus mainly on the effect of RIN when working at high optical power values, which is the case for IM-DD systems targeting data rates of 400 Gb/s and beyond.

Next, we have that the output optical power of the modulator can be described as
\begin{align} \label{eq:P_tx1}
    P_{\mathrm{tx}}(t) &= P_{\mathrm{cw}}(t)f\left(u(t)\right)
\end{align}
where $f(\cdot)$ is the nonlinear transfer function of the E/O modulator. Common modulators used in high-speed IM-DD systems are electro-absorption modulators, thin-film lithium niobate modulators, and Mach-Zehnder modulators \cite[Sec.~II]{che2023modulation}. As an example, the red curve in Fig.~\ref{fig:mzm_tf} shows the transfer function $f(\cdot)$ of a Mach-Zehnder modulator (MZM). From this figure, it can be observed that the optical power transfer function has a periodicity given by the MZM parameter $V_\pi$ [V]. As mentioned previously, to deal with the nonlinearity of $f(\cdot)$, we consider that the target optical PAM power levels are mapped to the input voltage levels by the inverse modulator transfer function $f^{-1}(\cdot)$ \cite{liang2023geometric,yang2023digital}. By doing so, the pre-distortion function $g(\cdot)$ from \eqref{eq:U_k} is set to
\begin{equation} \label{eq:u_k}
    g(X_k) = f^{-1}\left(\frac{\eta}{\Pcw}(X_k+\beta)\right),
\end{equation}
where $\eta$ is the electro-optical conversion factor in [W/V], and \mbox{$\beta\geq |\min\{\mathcal{X}\}|$} is a bias required to satisfy the nonnegativity constraint for IM and where $\mathcal{X}$ is the set of bipolar constellation points.
\begin{figure}[!t]
    \centering
    \input{FigTikz/MZM_TF.tex}
    \caption{Mach-Zehnder modulator transfer function. The difference between power levels at the output is defined as the optical modulation amplitude (OMA).}
    \label{fig:mzm_tf}
\end{figure}
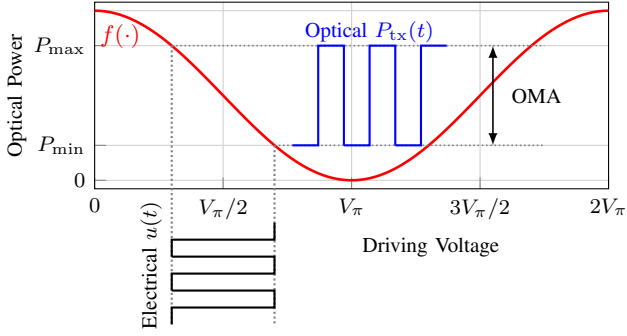

At the output of the modulator, we define the two outermost power levels $P_{\mathrm{max}}$ and $P_{\mathrm{min}}$ as
\begin{equation} \label{eq:P_max}
    P_{\mathrm{max}}\triangleq\mathbb{E}[\max\{P_{\mathrm{tx}}(kT)\}]=\eta(\max\{\mathcal{X}\}+\beta),
\end{equation}
\begin{equation} \label{eq:P_min}
    P_{\mathrm{min}}\triangleq\mathbb{E}[\min\{P_{\mathrm{tx}}(kT)\}]=\eta(\min\{\mathcal{X}\}+\beta),
\end{equation}
which correspond to the maximum and minimum transmitted optical power levels that carry information from the constellation $\mathcal{X}$. This corresponds to the amplitudes of the signal $P_{\mathrm{tx}}(t)$ at the time instants $t=kT$.\footnote{The definitions for $P_{\mathrm{max}}$ and $P_{\mathrm{min}}$, and the inclusion of expectations in those definitions will be explained further in Sec.~\ref{sec:receiver}.} With \eqref{eq:P_max} and \eqref{eq:P_min}, the optical signal is characterized by the following two metrics: the optical modulation amplitude (OMA) \cite{yu2014trade} defined as
\begin{equation} \label{eq:oma}
    \mathrm{OMA} \triangleq P_{\mathrm{max}} - P_{\mathrm{min}}=\eta\left(\max\{\mathcal{X}\}-\min\{\mathcal{X}\}\right),
\end{equation}
and the extinction ratio $\varepsilon_r>1$ \cite{che2023modulation} defined as
\begin{equation} \label{eq:er}
    \varepsilon_r \triangleq \frac{P_{\mathrm{max}}}{P_{\mathrm{min}}} =\frac{\max\{\mathcal{X}\}+\beta}{\min\{\mathcal{X}\}+\beta}.
\end{equation}
Hence, for given values of $\mathrm{OMA}$ and $\varepsilon_r$, the parameters $\eta$ and $\beta$ can be calculated directly from \eqref{eq:oma} and \eqref{eq:er}. The OMA of the optical signal is shown in Fig.~\ref{fig:mzm_tf} (blue curve) as the peak-to-peak value of the optical signal, and is measured in units of [mW] or [dBm]. Also, note that for a given $\varepsilon_r$ the value of $\beta$ will change accordingly with $\min\{\mathcal{X}\}$ and $\max\{\mathcal{X}\}$. Furthermore, for a finite $\varepsilon_r$, then $\beta >|\min\{\mathcal{X}\}|$, and in the \textit{ideal} case where $\varepsilon_r\to\infty$ then $\beta=|\min\{\mathcal{X}\}|$. Thus, the constraint $\beta \geq |\min\{\mathcal{X}\}|$ is always satisfied.

\subsection{Optical Channel}
The optical signal propagates through a fiber of length $L$. Only the attenuation of the signal is considered since the wavelength of the O-band laser is in the zero dispersion regime \cite{alam2021net}, and thus, chromatic dispersion is negligible. Therefore, after propagation, the received optical power is
\begin{equation} \label{eq:Prx}
    P_{\mathrm{rx}}(t) = \varphi P_{\mathrm{tx}}(t),
\end{equation}
where $10\log_{10}(1/\varphi) = \alpha L$, and $\alpha$ is the attenuation coefficient in [dB/km].

\subsection{Receiver} \label{sec:receiver}

As shown in Fig.~\ref{fig:imdd_link}, a single PD with responsivity $\mathfrak{R}$ [A/W] is used to convert the incident optical power $P_{\mathrm{rx}}(t)$ into an electrical current. Due to this conversion process, electrical noise (also referred to as thermal noise) $n_{\thn}(t)$ is generated. Thermal noise is modeled as additive white Gaussian noise with zero mean and PSD \mbox{$S_{\thn}(f)=N_0/2$}. The resulting photocurrent is then given by
\begin{equation} \label{eq:I_pd}
    i(t) = \mathfrak{R}P_\mathrm{rx}(t) + n_{\thn}(t).
\end{equation}
The TIA input is AC coupled, and as a result, the bias of the photocurrent $\beta^{\prime}$ is removed. The resulting current is then amplified and converted to a voltage $v(t)$ via a gain $G$ (with units of [$\Omega$]). The resulting voltage signal is
\begin{align} \label{eq:v_1}
    v(t) &= G(i(t) - \beta^{\prime})\\ \label{eq:v_2}  
    &= G\mathfrak{R}P_{\mathrm{rx}}(t) - G\beta^{\prime}+Gn_\thn(t)\\ \label{eq:tia_voltage}
    &= G\mathfrak{R}\varphi P_{\mathrm{tx}}(t) - G\beta^{\prime}+Gn_\thn(t),
\end{align}
where to pass from \eqref{eq:v_1} to \eqref{eq:v_2} we use \eqref{eq:I_pd}, and to pass from \eqref{eq:v_2} to \eqref{eq:tia_voltage} we use \eqref{eq:Prx}. The gain of the TIA is set to 
\begin{equation} \label{eq:gain}
    G=\frac{1}{\mathfrak{R} \varphi \eta}
\end{equation}
to compensate for the PD responsivity, fiber losses, and electro-optical conversion. Furthermore, the photocurrent bias to be removed is set to $\beta^{\prime}=G^{-1}\beta$.

Lastly, the output voltage of the amplifier \eqref{eq:tia_voltage} is filtered with a matched filter $p(-t)$ with bandwidth $B$ to band-limit the noises, and then, digitized by taking samples every $t=kT$ seconds. This processing results in the discrete channel output samples $Y_k=\tilde{v}(kT)$\footnote{We use $\tilde{x}(t)$ to denote the filtered version of the signal $x(t)$.} given by
\begin{align} \label{eq:Y_k_prev}
    Y_k &= \eta^{-1}\tilde{P}_{\mathrm{tx}}(kT) - \beta+G\tilde{n}_\thn(kT)\\ \label{eq:Y_k}
    &= X_k + (X_k+\beta)\tilde{n}_\rin(kT) + G\tilde{n}_\thn(kT),
\end{align}
where the derivation for \eqref{eq:Y_k} is given in Appendix~\ref{app:A}. Furthermore, the same derivation is used to obtain the relationship between $\mathcal{X}$ and the power levels $P_{\mathrm{max}}$ and $P_{\mathrm{min}}$ from \eqref{eq:P_max} and \eqref{eq:P_min}.

The entire IM-DD system can be modeled as the channel expression given at the bottom of Fig.~\ref{fig:imdd_link}. The derivation for this equivalent channel model is described next.

\subsection{Equivalent Channel Model}

We assume that the effect of previously transmitted symbols can be neglected with proper equalization techniques \cite[Eq.~(15)]{che2023modulation}. Hence, we model the IM-DD system as a memoryless channel with additive noise and channel output given by
\begin{equation} \label{eq:ch_eq1}
    Y_k = X_k + Z_k^\prime,
\end{equation}
where $X_k$ are the transmitted PAM symbols. Since in \eqref{eq:Y_k_prev} both $\tilde{n}_\rin(kT)$ and $\tilde{n}_\thn(kT)$ are Gaussian and independent random variables (RVs), the effect of both noises can be grouped into a single RV $Z_k^\prime$. The additive noise $Z^\prime$ is signal-dependent as seen from \eqref{eq:Y_k}, i.e., $Z^\prime=Z\!\cdot\!\sigma(X)$, where $Z$ is a zero-mean and unit variance Gaussian RV. The effect of RIN and thermal noise translate into a signal-dependency given by
\begin{equation} \label{eq:ch_noise1}
    \sigma(X) = \sqrt{\sigma_{\thn}^2 + (X+\beta)^2\sigma_{\rin}^2},
\end{equation}
where $\sigma_{\thn}^2$ is the variance contribution due to thermal noise and $(X+\beta)^2\sigma_{\rin}^2$ is the contribution associated to RIN. Combining \eqref{eq:ch_eq1} and \eqref{eq:ch_noise1}, we have the equivalent channel
\begin{equation} \label{eq:ch_eq1_full}
    Y = X + Z\sqrt{\sigma_\thn^2 + (X+\beta)^2\sigma_\rin^2},
\end{equation}
where for an electrical system bandwidth $B$, the variance of the receiver noise and RIN is calculated as \cite[Eq.~(17)]{xu2011impact}
\begin{equation} \label{eq:sigmas}
    \sigma_\thn^2 = G^2 (N_0/2) B\, ,\quad \sigma_\rin^2 = 10^{\frac{\mathrm{RIN}}{10}} B.
\end{equation}

%% file: FigTikz/block_diagram_IMDD2.tex
\definecolor{ashgrey}{rgb}{0.75, 0.75, 0.75}
\definecolor{antiquebrass}{rgb}{0.98, 0.81, 0.69}
\definecolor{brilliantlavender}{rgb}{0.96, 0.73, 1.0}
\definecolor{bgcolor}{rgb}{0.97, 0.98, 0.96}
\definecolor{bgcolor2}{rgb}{0.99, 0.93, 0.89}
\newcommand{\OColor}{blue!80!black}

\tikzstyle{block} = [draw, line width = 1pt, fill=black!20, rectangle, minimum height=30pt, rounded corners=0.1cm, text width=2.5em,align=center]

\tikzstyle{block_wide} = [draw, line width = 1pt, fill=black!20, rectangle, minimum height=30pt, rounded corners=0.1cm, text width=3.5em,align=center]

\tikzstyle{block_wide2} = [draw, line width = 1pt, fill=black!20, rectangle, minimum height=20pt, rounded corners=0.1cm, text width=3em,align=center]

\tikzstyle{block_wide3} = [draw, line width=1pt, fill=black!20, rectangle, minimum height=40pt, rounded corners=0.1cm, text width=7.0em,align=center]

\tikzstyle{block2} = [draw, line width = 1pt, fill=black!20, rectangle, minimum height=30pt, minimum width=30pt, rounded corners=0.1cm, text width=5em,align=center]

\tikzstyle{Cir} = [draw, circle,  minimum size=2.15em]
\tikzstyle{circ} = [circle, draw, minimum size=10pt, text centered,inner sep=0pt]

\newcommand{\shotNoise}{0}  
\newcommand{\lw}{0.7}       

\begin{tikzpicture}[auto, node distance=1 cm,>=to,line width=\lw]

    \node [coordinate] (input) {};  
    \node [block, right = 2em of input, fill = gray!30] (DAC) {$g(\cdot)$};
    \node [block, right = 2.5em of DAC, fill = gray!30] (pulse) {Filter $p(t)$};
    
    \node [below=0em of DAC](){$\beta$, $\eta$};
    \node [block,right = 2.5em of pulse, fill = red!20] (EO) {E/O};
    \node [below=0em of EO](){$f(\cdot)$};

    \node[block_wide3,dashed,above=0.9em of pulse,xshift=-16pt,fill=red!20] (laser) {};
    \node[left=0.7em of laser,rotate=90,xshift=1.5em](){Laser};
    \node[circ,above=1.3em of pulse,xshift=10pt](sum1){\large $+$};
    \node[left=1.5em of sum1](P_A){$\Pcw$};
    \node[above=0.2em of P_A](P_A_rin){$\Pcw\!\cdot\!n_{\mathrm{rin}}(t)$};

    \draw [\OColor,solid,line width=0.5pt,opacity=1] ($(EO.east)+(2.9,0.3)$) circle (3mm);
    \draw [\OColor,solid,line width=0.5pt,opacity=1] ($(EO.east)+(3.0,0.3)$) circle (3mm);
    \draw [\OColor,solid,line width=0.5pt,opacity=1] ($(EO.east)+(3.1,0.3)$) circle (3mm);
    \node[right=2.6cm of EO, yshift=-1.7em](fiber_mid){$\alpha L$};
    \node[right=0.7em of EO, yshift=0.8em](){$P_{\mathrm{tx}}(t)$};

    \node [block, right = 6cm of EO, fill = red!20] (OE) {PD};
    \node[left=0.9em of OE, yshift=0.8em](Prx){$P_{\mathrm{rx}}(t)$};
    \node[below=0em of OE](){$\mathfrak{R}$};

    \ifthenelse{\equal{\shotNoise}{1}}{
        \node[circ,fill=gray!20,right=1.5em of OE](sum2){\large $+$};
        \node[circ,fill=gray!20,right=1.5em of sum2](sum3){\large $+$};
        \node[below=1.8em of sum2](noiseshot){$\sqrt{\mathfrak{R}P_{\mathrm{rx}}(t)}n_{\mathrm{sh}}(t)$};
        \draw [draw,-latex] (noiseshot) -- (sum2);
        \draw [draw,-latex] (OE) -- (sum2);
        \draw [draw,-latex] (sum2) -- (sum3);
    }{
        \node[circ,right=1.5em of OE](sum3){\large $+$};
        \draw [draw,-latex] (OE) -- (sum3);
    }

    \node [above=1.8em of sum3](noise){$n_{\mathrm{ele}}(t)$};
    \node [block, right = 2.5em of sum3, fill = gray!30] (TIA) {TIA};
    \node [below=0em of TIA](){$G$};
    \node [block, right = 2.5em of TIA, fill = gray!30] (ADC) {Filter $p(-t)$};
    \node [below=0em of ADC](){$B$};
    \node [right=-0.1em of ADC,yshift=8pt](){$\tilde{v}(t)$};
    \node [right=2.0em of ADC] (sampling) {};
    \node [below=0.4em of sampling, xshift=0.8em] () {\small{$t=kT$}};
    \node [right = 0.8em of sampling] (output) {};

    \draw [draw,-latex] (input) -- node[left,text width=2.2em, align = left]{$X_k$}(DAC);
    \draw [draw,-latex] (DAC) -- node[midway,above]{$U_k$}(pulse);
    \draw [draw,-latex] (pulse) -- node[midway,above]{$u(t)$}(EO);
    
    \draw [draw,-latex, color=\OColor] (sum1) -| node[midway,right,xshift=-3.5em,yshift=0.8em]{\textcolor{black}{$P_{\mathrm{cw}}(t)$}}(EO);
    \draw [draw,-latex, color=\OColor] ($(P_A)+(1.0em,0)$) -- (sum1);
    \draw [draw,-latex, color=\OColor] ($(P_A_rin)+(2.5em,0)$) -| (sum1);
    
    \draw [draw,-latex, color=\OColor] (EO) -- node[midway,below]{\textcolor{black}{Fiber}}(OE);
    \draw [draw,-latex] (noise) -- (sum3);

    \draw [draw,-latex] (sum3) -- node[midway,above]{$i(t)$}(TIA);
    \draw [draw,-latex] (TIA) -- node[midway,above]{$v(t)$}(ADC);
    \draw[draw] (ADC) -- (sampling);
    \draw [draw] ($(sampling.west)+(0.2em,1em)$) -- ($(output.west)$);
    \draw [draw,-latex] ($(output.west)$) --  node[right,text width=2.4em, align = right]{$Y_k$}($(output.west)+(3em,0)$);

    \draw[->, >=latex, densely dashed, gray, bend right=20] ($(sampling)+(0.6em,1.2em)$) to ($(sampling)+(0.6em,-0.8em)$);


    \node[below=5em of Prx,xshift=-2.2em](ch_eq){$Y_k=X_k+Z_k\sqrt{\sigma_{\mathrm{ele}}^2 + (X_k+\beta)^2\sigma_{\mathrm{rin}}^2}$};
    \draw[draw,-stealth,dotted,thick,color=gray] ($(DAC.south)+(-1.05,1.5em)$) |- (ch_eq);
    \draw[draw,-stealth,dotted,thick,color=gray] ($(ADC.south)+(7.0em,1.5em)$) |- (ch_eq);

    \begin{pgfonlayer}{background}
        \draw[dotted,fill=bgcolor,rounded corners=2pt] ($(DAC.west)+(-0.25,-1.1)$) rectangle ($(EO.east)+(+0.15,2.4)$);     
        \draw[dotted,fill=bgcolor,rounded corners=2pt] ($(EO.east)+(+0.30,-1.1)$) rectangle ($(OE.west)+(-0.40,1.6)$);     
        \draw[dotted,fill=bgcolor,rounded corners=2pt] ($(OE.west)+(-0.25,-1.1)$) rectangle ($(ADC.east)+(+1.7,1.6)$);     

        \node[below=1.5em of pulse](){Transmitter};
        \node[below=1.5em of TIA](){Receiver};
        \node[below=0.6em of fiber_mid](){Optical Channel};
    \end{pgfonlayer}
    

    \node[above = 2.5cm of input, xshift=-1.5em] (elec_0){};
    \node[above = 2.5cm of input, xshift=1em] (elec_1){};
    \node[right = 1.3cm of elec_1, xshift=-1.5em] (opt_0){};
    \node[right = 1.3cm of elec_1, xshift=1em] (opt_1){};
    
    \draw[draw,-latex] (elec_0) -- node[right,text width=4.3em, align=left,xshift=0.6em]{\footnotesize Electrical}(elec_1);
    \draw[draw,-latex,\OColor] (opt_0) -- node[right,text width=3.5em, align=left,xshift=0.6em]{\footnotesize \textcolor{black}{Optical}}(opt_1);
 
\end{tikzpicture}

%% file: FigTikz/RIN_example3.tex
\begin{tikzpicture}[scale=1]

    \draw[draw,-latex] (-0.2,0) -- (4.8,0) node[right] {\scriptsize $t$};
    \draw[draw,-latex] (0,-0.2) -- (0,1.8) node[above] {\scriptsize $P_{\mathrm{cw}}(t)$};;
    \node[rotate=90]() at (-0.55,0.8)  {\scriptsize Optical Power};
    
    \draw[densely dotted, gray] (0,0.4) -- (4.5,0.4);
    \node[blue!80!black]() at (1.2,0.62) {\scriptsize $P_{1}=1$ mW};
    \draw[blue!80!black, thin] plot file{FigTikz/data_txt/laser_noise_example3.txt};
    \draw[red, thin, domain=-0.1:0.1, samples=50, rotate=90, xshift=11.5,yshift=0.18] plot (\x,{0.35*exp(-\x*\x/(2*4e-4))});
    
    \draw[densely dotted, gray] (0,1.2) -- (4.5,1.2);
    \node[blue!50!black]() at (3.5,1.48) {\scriptsize $P_{2}=3$ mW};
    \draw[blue!50!black, thin] plot file{FigTikz/data_txt/laser_noise_example4.txt};
    \draw[red, thin, domain=-0.3:0.3, samples=50, rotate=90, xshift=34,yshift=0.18] plot (\x,{0.35*exp(-\x*\x/(2*0.0036))});

\end{tikzpicture}

%% file: FigTikz/noise_comparison.tex
\definecolor{my_green}{rgb}{0.55, 0.71, 0.0}    
\definecolor{my_purple}{rgb}{0.5, 0, 0.5}

\begin{tikzpicture}
    \begin{semilogyaxis}[
    width=0.95\columnwidth,  
    height=2.4in,
    font=\footnotesize,
    xmin=-20, xmax=10,
    ymin=8e-14, ymax=1e-7,
    xlabel={Transmitted Optical Power [dBm]},
    ylabel={Noise Variance $[W]$},
    ytick pos=left,
    xtick pos=bottom,
    grid=major,
    grid style = {solid,lightgray!75},
    legend style = {legend pos=south east, font=\scriptsize, legend cell align=left, row sep=-0.5ex},
    ytickten={-14,-13,...,8},
    ]
        \addplot[name path=lregion, green!50!black, opacity=0.5, forget plot] coordinates{(-12,1e-15) (-12,1e-7)};
        \addplot[name path=rregion, green!50!black, opacity=0.5, forget plot] coordinates{(6,1e-15) (6,1e-7)};
        \addplot[green!30, opacity=0.3, forget plot] fill between[of=lregion and rregion];

        \addplot[color=blue, solid, line width=1pt] table[x={ROP},y={s2th}] {FigTikz/data_txt/noise_dependency3.txt};
        \addlegendentry{Thermal Noise}
        
        \addplot[color=my_purple, solid, line width=1pt] table[x={ROP},y={s2shot}] {FigTikz/data_txt/noise_dependency3.txt};
        \addlegendentry{Shot Noise}
        
        \addplot[color=red, solid, line width=1pt] table[x={ROP},y={s2rin}] {FigTikz/data_txt/noise_dependency3.txt};
        \addlegendentry{RIN}
        
        \addplot[color=black, dashed, line width=1pt] table[x={ROP},y={s2total}] {FigTikz/data_txt/noise_dependency3.txt};
        \addlegendentry{Total Noise}

    \end{semilogyaxis}

\end{tikzpicture}

%% file: FigTikz/MZM_TF.tex
\definecolor{my_green}{rgb}{0.55, 0.71, 0.0}
\definecolor{my_purple}{rgb}{0.5, 0, 0.5}

\pgfmathsetmacro{\xb}{0.3}
\pgfmathsetmacro{\xe}{0.7}
\pgfmathsetmacro{\yb}{0.5+0.5*cos(180*\xb)}
\pgfmathsetmacro{\ye}{0.5+0.5*cos(180*\xe)}
\pgfmathsetmacro{\sb}{-0.85}
\pgfmathsetmacro{\ssb}{0.77}

\begin{tikzpicture}
    \begin{axis}[
    width=0.95\linewidth,  
    height=1.6in,
    xmin=0, xmax=2,
    ymin=-0.05, ymax=1.05,
    xlabel={Driving Voltage},
    xlabel style={xshift=30pt},
    ylabel={Optical Power},
    ylabel style={yshift=-4pt},
    xtick={0, 0.5, 1, 1.5, 2},
    xtick pos=bottom,
    ytick pos=left,
    xticklabels={$0$,$V_\pi/2$,$V_\pi$, $3V_\pi/2$, $2V_\pi$},
    ytick={0,\ye,\yb,1},
    yticklabels={$0$,$P_{\mathrm{min}}$,$P_{\mathrm{max}}$,},
    grid=major,
    grid style = {lightgray!75},
    font=\footnotesize,
    clip=false,
    ]
   
        \addplot[domain=0:2, samples=100, color=red, solid, line width=1.0pt]{0.5+0.5*cos(180*x)};


        \draw[gray,densely dotted, thick] (axis cs:\xb,\yb) -- (axis cs:\xb,-0.85);
        \draw[gray,densely dotted, thick] (axis cs:\xe,\ye) -- (axis cs:\xe,-0.85);

        \draw[black, line width=0.8pt] (axis cs:\xb,\sb) -- (\xb,\sb+0.1) -- (\xe,\sb+0.1) -- (\xe,\sb+0.2) -- (\xb,\sb+0.2);
        \draw[black, line width=0.8pt] (axis cs:\xb,\sb+0.2) -- (\xb,\sb+0.3) -- (\xe,\sb+0.3) -- (\xe,\sb+0.4) -- (\xb,\sb+0.4);
        \draw[black, line width=0.8pt] (axis cs:\xb,\sb+0.4) -- (\xb,\sb+0.5) -- (\xe,\sb+0.5) -- (\xe,\sb+0.6);
        
        \node[rotate=90] (RF) at (axis cs:\xb-0.08,\sb+0.35) {\footnotesize Electrical $u(t)$};    
        
        \draw[gray,densely dotted, line width=0.6pt] (axis cs:\xb,\yb) -- (1.75,\yb);
        \draw[gray,densely dotted, line width=0.6pt] (axis cs:\xe,\ye) -- (1.75,\ye);

        \draw[blue, line width=0.8pt] (axis cs:\ssb,\ye) -- (\ssb+0.1,\ye) -- (\ssb+0.1,\yb) -- (\ssb+0.2,\yb) -- (\ssb+0.2,\ye);
        \draw[blue, line width=0.8pt] (axis cs:\ssb+0.2,\ye) -- (\ssb+0.3,\ye) -- (\ssb+0.3,\yb) -- (\ssb+0.4,\yb) -- (\ssb+0.4,\ye);
        \draw[blue, line width=0.8pt] (axis cs:\ssb+0.4,\ye) -- (\ssb+0.5,\ye) -- (\ssb+0.5,\yb) -- (\ssb+0.6,\yb);

        \node[blue] (Opt) at (axis cs:\ssb+0.3,\yb+0.08) {\footnotesize Optical $P_{\mathrm{tx}}(t)$};

        \draw[>=latex,<->,line width=0.6] (1.55,\ye) -- (1.55,\yb);
        \node[] (OMA) at (axis cs:1.72,0.5) {\footnotesize OMA};
        \node[red] (function) at (axis cs:0.1,0.85) {\small $f(\cdot)$};

    \end{axis}

\end{tikzpicture}

%% file: sections/3.symbol_error.tex
\section{Symbol Error Probability Analysis}
\label{sec:symbol_error}

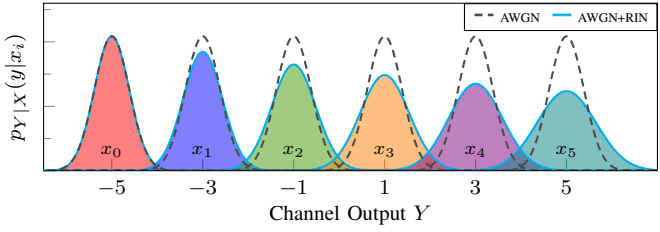
\begin{figure}[!t]
    \centering
    \input{FigTikz/noise_variance2.tex}
    \caption{Conditional PDF of the channel output $y$ given a transmitted PAM symbol $X=x_i$.}
    \label{fig:cond_pdf}
\end{figure}

From now on, we use $x_i$, with $i=0,1,\dotsc,M\!-\!1$, to represent the $i$-th symbol of the PAM-$M$ constellation. For the memoryless model \eqref{eq:ch_eq1_full} under consideration, the conditional probability density function (PDF) $p_{Y|X}(y|x_i)$ is given by
\begin{equation} \label{eq:cond_pdf}
    p_{Y|X}(y|x_i) = \frac{1}{\sqrt{2\pi\sigma_i^2}}\mathrm{exp}\left(-\frac{(y-x_i)^2}{2\sigma_i^2}\right),
\end{equation}
where $\sigma_i^2$ is the noise variance conditioned on $X=x_i$ that can be obtained from \eqref{eq:ch_eq1_full} as
\begin{equation} \label{eq:cond_var}
    \sigma_i^2 = \sigma_\thn^2 + (x_i+\beta)^2\sigma_\rin^2.
\end{equation}
Since $(x_i+\beta)\!\geq\! 0$, and $x_j\!>\!x_i$, then $\sigma_j\!>\!\sigma_i$ for all $j\!>\!i$. An example of the conditional PDF \eqref{eq:cond_pdf} is visualized in Fig.~\ref{fig:cond_pdf} where we consider a PAM-6 constellation with $\mathcal{X} =\{\pm1, \pm3, \pm5\}$. Due to the presence of RIN, symbols mapped to larger optical power values result in a larger noise variance, which can be seen as a wider distribution with respect to the case with AWGN only. In this paper, we consider the optimal decision thresholds of a MAP detector. Thus, in order to capture the effect of the signal-dependent noise, the optimal threshold calculation will be developed in the following section.

\begin{table*}[!t]
    \centering
    \caption{Comparison of optimal thresholds for different cases}
    \begin{tabular}{c|c|c}
        \hline\hline
        Case & Optimal thresholds $r_i^\star$ & Expression \\
        \hline
        General, Exact & 
        $\frac{1}{\sigma_j^2-\sigma_i^2}\left(\sigma_j^2x_i-\sigma_i^2x_j+\sigma_i\sigma_j\left[ (x_j-x_i)^2+2(\sigma_j^2-\sigma_i^2)\log\left(\frac{P_X(x_i)}{P_X(x_j)}\frac{\sigma_j}{\sigma_i}\right) \right]^{\frac{1}{2}}\right)$&
        \eqref{eq:r_opt}, Theorem~\ref{theo:optimal}\\

        Uniform, Exact &
        $\frac{1}{\sigma_j^2-\sigma_i^2}\left(\sigma_j^2x_i-\sigma_i^2x_j+\sigma_i\sigma_j\left[ (x_j-x_i)^2+2(\sigma_j^2-\sigma_i^2)\log\left(\frac{\sigma_j}{\sigma_i}\right) \right]^{\frac{1}{2}}\right)$&
        \eqref{eq:r_opt_alt}, Corollary~\ref{cor:exact}, \cite{chagnon2014experimental}\\

        Uniform, Approx. &
        $\frac{x_i\sigma_j + x_j\sigma_i}{\sigma_i+\sigma_j}$ &
        \eqref{eq:map_approx}  Corollary~\ref{cor:approx}, \cite[Ch.~4.6.1]{agrawal2012fiber}\\

        AWGN Only, Exact &
        $\frac{\sigma^2}{x_j-x_i}\log\left(\frac{P_X(x_i)}{P_X(x_j)}\right)+\frac{x_i+x_j}{2}$ &
        $\sigma_i^2=\sigma_j^2$, $\forall i,j$ \cite[Ch.~4.1]{proakis2008digital}\\
        
        \hline\hline
    \end{tabular}
    
    \label{tab:r_comparison}
\end{table*}

\subsection{MAP Decision Thresholds}
The optimal decision criteria for the noisy received symbols to minimize SER is the MAP detection rule given by
\begin{align}
    \hat{x}&=\underset{x_i\in\mathcal{X}}{\arg\!\max}\ P_{X|Y}(x_i|y)\\
    &= \underset{x_i\in\mathcal{X}}{\arg\!\max}\ P_X(x_i)p_{Y|X}(y|x_i), \label{eq:map_argmax0}
\end{align}
where $P_X(x)$ is the probability mass function of the input symbols.

Let $r_i$ be the decision threshold used to decide the transmitted symbol based on the channel observation $y$. The optimal decision threshold is denoted by $r_i^\star$ and satisfies
\begin{equation} \label{eq:MAP_eq0}
    P_X(x_i)p_{Y|X}(r_i|x_i) = P_X(x_j)p_{Y|X}(r_i|x_j),
\end{equation}
where the subindex $j=i+1$ has been used for ease of notation. The following theorem gives an explicit solution for the optimal decision thresholds. The expression in the following theorem is, to the best of our knowledge, not known in the literature.
\begin{theorem}[\textit{Optimal Decision Thresholds}] \label{theo:optimal}
The optimal decision thresholds $r_i^\star$ for the channel in \eqref{eq:ch_eq1_full}, are
\begin{equation} \label{eq:r_opt}
    \begin{split}
        r_i^\star = &\frac{1}{\sigma_j^2 - \sigma_i^2}\Bigg( \sigma_j^2 x_i - \sigma_i^2 x_j + \sigma_i\sigma_j\Bigg[(x_j-x_i)^2\Bigg.\Bigg.\\
        &\quad \quad \;\Bigg.\Bigg.+ 2(\sigma_j^2-\sigma_i^2)\log\left(\frac{P_X(x_i)}{P_X(x_j)}\frac{\sigma_j}{\sigma_i}\right)\Bigg]^{\frac{1}{2}} \Bigg).
    \end{split}
\end{equation}
\end{theorem}
\begin{proof}
    See Appendix~\ref{app:B}.
\end{proof}

From \eqref{eq:r_opt}, we see that $r_i^\star$ explicitly depends on the variance of the neighboring symbols ($\sigma_i^2$ and $\sigma_j^2$), as well as their probabilities ($P_X(x_i)$ and $P_X(x_j)$), thereby capturing the signal-dependency noise.

Theorem~\ref{theo:optimal} can be used for the particular case of a uniform input distribution $P_X(x)$. In this scenario, the thresholds depend only on the conditional variances and the symbol locations. The following corollaries are obtained from Theorem~\ref{theo:optimal} and matches results in the literature \cite{chagnon2014experimental}.

\begin{corollary}[\textit{Equally-likely Symbols: Exact}, \cite{chagnon2014experimental}] \label{cor:exact}
    For an uniform input distribution, the optimal decision thresholds are
    \begin{equation} \label{eq:r_opt_alt}
    \begin{split}
        r_i^{\star} = &\frac{1}{\sigma_j^2 - \sigma_i^2}\Bigg( \sigma_j^2 x_i - \sigma_i^2 x_j + \sigma_i\sigma_j\Bigg[(x_j-x_i)^2\Bigg.\Bigg.\\
        &\quad\quad\quad\quad\quad \;\Bigg.\Bigg.+ 2(\sigma_j^2-\sigma_i^2)\log\left(\frac{\sigma_j}{\sigma_i}\right)\Bigg]^{\frac{1}{2}} \Bigg).
    \end{split}
\end{equation}
\end{corollary}
\begin{proof}
    Setting $P_X(x_i)=P_X(x_j)$ inside $\log(\cdot)$ of \eqref{eq:r_opt}.
\end{proof}
The result of equation~\eqref{eq:r_opt_alt} is known in the literature (see, e.g.,  \cite{chagnon2014experimental}). This expression is often further approximated into a simplified expression where of the terms from $r_i^{\star}$ is neglected. This approximation is given by the following Corollary, whose proof again follows directly from Theorem~\ref{theo:optimal}.

\begin{corollary}[\textit{Equally-likely Symbols: Approximation}, \cite{agrawal2012fiber}] \label{cor:approx}
    For an uniform input distribution, the optimal decision thresholds can be approximated by
\begin{equation} \label{eq:map_approx}
    r_i^{\star} \approx \hat{r}_i = \frac{x_i\sigma_j+x_j\sigma_i}{\sigma_i+\sigma_j}.
\end{equation}
\end{corollary}
\begin{proof}
    From \eqref{eq:r_opt_alt}, setting the term $2(\sigma_j^2-\sigma_i^2)\log\left(\frac{\sigma_j}{\sigma_i}\right)\approx 0$, and then simplifying the expression. 
\end{proof}
 
The expression in Corollary~\ref{cor:approx} is a common approximation employed in the literature (e.g., \cite[Ch.~4.6.1]{agrawal2012fiber}). Although this approximation results in a suboptimal threshold calculation, the approximation behaves well for the uniform case. In contrast, when the input distribution is nonuniform, the SER calculation is less accurate. We will show this in Section \ref{sec:results}.

Table~\ref{tab:r_comparison} shows a comparison of the optimal decision thresholds for different scenarios. In this table, the well-known decision thresholds for the case of the pure AWGN channel ($\sigma_i^2=\sigma_j^2$, for all $j,i$) are also included (see, e.g., \cite[Ch.~4.1]{proakis2008digital}). Note that the AWGN only expression can also be obtained by setting $a=0$ in \eqref{eq:r_abcd} of Appendix~\ref{app:B}, and then solving for $r_i$.

\subsection{PAM Symbol Error Probability}
At the receiver, the conditional error probability for a given transmitted symbol $x_i$ is defined as \mbox{$P_e^{(i)}\triangleq\Pr{Y\notin [r_{i-1},r_i]\mid X=x_i}$}, where $r_i$ is the $i$-th decision threshold (e.g., $r_i^\star$ in \eqref{eq:r_opt} or $\hat{r}_i$ in $\eqref{eq:map_approx}$). If noise pushes the signal beyond the thresholds $r_{i-1}$, $r_i$, there will be a symbol error. The probability of a symbol error for a given transmitted symbol $x_i$ is
\begin{align}
        P_e^{(i)} &= \Pr{Y\!<\!r_{i-1}|X\!=\!x_i} + \Pr{Y\!>\!r_i|X\!=\!x_i}\\
        &= \int_{\!-\infty}^{r_{i-1}}  \!\! p_{Y|X}(y|x_i)\mathrm{d}y + \!\int_{r_i}^\infty  \!\!p_{Y|X}(y|x_i)\mathrm{d}y \\
        &=  Q\left( \frac{x_i - r_{i-1}}{\sigma_i} \right) + Q\left( \frac{r_i - x_i}{\sigma_i} \right),
\end{align}
where $Q(x) \triangleq \frac{1}{\sqrt{2\pi}}\int_x^\infty \mathrm{exp}\left(-\frac{\nu^2}{2}\right)\mathrm{d}\nu$ is the Gaussian Q-function. Then, the average symbol error probability $P_e$ is 
\begin{equation} \label{eq:SER_analytical}
    \begin{split}
        P_e &= \sum_{i=0}^{M-1}P_X(x_i)P_e^{(i)}\\
        &= \sum_{i=0}^{M-1}P_X(x_i)\left[ Q\left( \frac{x_i - r_{i-1}}{\sigma_i} \right) + Q\left( \frac{r_i-x_i}{\sigma_i} \right)\right],
    \end{split}
\end{equation}
with $r_{-1}=-\infty$ and $r_{M-1}=\infty$. The key contribution with respect to previous PAM SER expressions, such as the ones in \cite{chagnon2014experimental} or \cite{agrawal2012fiber}, is that, in addition to considering an arbitrary constellation $\mathcal{X}$, the input probability distribution $P_X(x)$ is also taken into account in both the thresholds calculation \eqref{eq:r_opt} and symbol error probability \eqref{eq:SER_analytical}. 

%% file: FigTikz/noise_variance2.tex
\definecolor{my_green}{rgb}{0.25, 0.6, 0.0}    
\definecolor{my_purple}{rgb}{0.5, 0, 0.5}
\definecolor{azure}{rgb}{0, 0.5, 0.5}

\newcommand{\lw}{0.8pt}

\newcommand{\graysolid}{\raisebox{2pt}{\tikz{\draw[color=black!50,solid,line width=\lw,mark=o,mark options={solid}](0,0) -- (2.5mm,0);}}}   
\newcommand{\redsolid}{\raisebox{2pt}{\tikz{\draw[color=red,solid,line width=\lw,mark=o,mark options={solid}](0,0) -- (2.5mm,0);}}}   
\newcommand{\bluesolid}{\raisebox{2pt}{\tikz{\draw[color=blue,solid,line width=\lw,mark=o,mark options={solid}](0,0) -- (2.5mm,0);}}}   
\newcommand{\purplesolid}{\raisebox{2pt}{\tikz{\draw[color=my_purple,solid,line width=\lw,mark=o,mark options={solid}](0,0) -- (2.5mm,0);}}}   
\newcommand{\greensolid}{\raisebox{2pt}{\tikz{\draw[color=my_green,solid,line width=\lw,mark=o,mark options={solid}](0,0) -- (2.5mm,0);}}}   

\newcommand{\blacksdashed}{\raisebox{2pt}{\tikz{\draw[color=black!70,densely dashed,line width=1.0pt,mark=o,mark options={solid}](0,0) -- (2.5mm,0);}}}   
\newcommand{\cyansolid}{\raisebox{2pt}{\tikz{\draw[color=cyan,solid,line width=1.0pt,mark=o,mark options={solid}](0,0) -- (2.5mm,0);}}}   

\begin{tikzpicture}
    \begin{axis}[
    width=1.1\columnwidth,  
    height=1.5in,
    xmin=-6.5, xmax=7,
    ymin=1e-3, ymax=1.3,    
    xlabel={Channel Output $Y$},
    x label style={yshift=3pt},
    xtick={-5,-3,...,5},
    yticklabels={\empty},
    ylabel={$p_{Y|X}(y|x_i)$},
    y label style={yshift=-5pt},
    ytick pos=left,
    xtick pos=bottom,
    font=\footnotesize,
    minor tick num=1,
    grid style = {dotted,lightgray},
    legend style = {legend pos=north east, font=\footnotesize, legend cell align=left, row sep=-0.5ex},
    domain=-7:7,
    samples=701,
    ]

        \pgfplotstableread{FigTikz/data_txt/histogram_rx_symbols_4a.txt}\datatable
        
                
        \addplot [name path=RIN0, cyan, solid, line width=0.8pt] {exp(-(x+5)^2/(2*(0.0407 + (-5+16.26)^2*8.3263e-4) ))/sqrt(2*pi*(0.0407 + (-5+16.26)^2*8.3263e-4))};
        \addplot [name path=RIN1, cyan, solid, line width=0.8pt] {exp(-(x+3)^2/(2*(0.0407 + (-3+16.26)^2*8.3263e-4) ))/sqrt(2*pi*(0.0407 + (-3+16.26)^2*8.3263e-4))};
        \addplot [name path=RIN2, cyan, solid, line width=0.8pt] {exp(-(x+1)^2/(2*(0.0407 + (-1+16.26)^2*8.3263e-4) ))/sqrt(2*pi*(0.0407 + (-1+16.26)^2*8.3263e-4))};
        \addplot [name path=RIN3, cyan, solid, line width=0.8pt] {exp(-(x-1)^2/(2*(0.0407 + (1+16.26)^2*8.3263e-4) ))/sqrt(2*pi*(0.0407 + (1+16.26)^2*8.3263e-4))};
        \addplot [name path=RIN4, cyan, solid, line width=0.8pt] {exp(-(x-3)^2/(2*(0.0407 + (3+16.26)^2*8.3263e-4) ))/sqrt(2*pi*(0.0407 + (3+16.26)^2*8.3263e-4))};
        \addplot [name path=RIN5, cyan, solid, line width=0.8pt] {exp(-(x-5)^2/(2*(0.0407 + (5+16.26)^2*8.3263e-4) ))/sqrt(2*pi*(0.0407 + (5+16.26)^2*8.3263e-4))};

        \foreach \sm in {-5,-3,...,5} {
            \addplot [black!70, densely dashed, line width=0.8pt] {exp(-(x-\sm)^2/(2*0.1463))/sqrt(2*pi*0.1463)};
        }
        
        \addplot[name path=baseline, forget plot] {1e-4};

        \addplot[red, opacity=0.5,forget plot] fill between[of=RIN0 and baseline];
        \addplot[blue, opacity=0.5,forget plot] fill between[of=RIN1 and baseline];
        \addplot[my_green, opacity=0.5,forget plot] fill between[of=RIN2 and baseline];
        \addplot[orange, opacity=0.5,forget plot] fill between[of=RIN3 and baseline];
        \addplot[my_purple, opacity=0.5,forget plot] fill between[of=RIN4 and baseline];
        \addplot[azure, opacity=0.5,forget plot] fill between[of=RIN5 and baseline];
        
        \node [draw,font=\tiny,fill=white,anchor=north east] at (rel axis cs:1,1) {\shortstack[l]{
        \blacksdashed \hspace{0.01cm} AWGN \hspace{0.025cm} \cyansolid \hspace{0.01cm} AWGN+RIN
        }};

        \node [font=\scriptsize] at (axis cs:-5,0.15) {$x_0$};
        \node [font=\scriptsize] at (axis cs:-3,0.15) {$x_1$};
        \node [font=\scriptsize] at (axis cs:-1,0.15) {$x_2$};
        \node [font=\scriptsize] at (axis cs:1,0.15) {$x_3$};
        \node [font=\scriptsize] at (axis cs:3,0.15) {$x_4$};
        \node [font=\scriptsize] at (axis cs:5,0.15) {$x_5$};

    \end{axis}

\end{tikzpicture}

%% file: sections/4.simulations.tex
\section{Numerical Simulations}
\label{sec:results}

In this section, the analytical expression \eqref{eq:SER_analytical} for the SER  presented in the previous section is compared with numerical simulations to validate the accuracy of the developed expression. The SER is evaluated for different values of OMA, by adjusting the parameter $\eta$ (see~\eqref{eq:oma}). By adjusting $\eta$, the TIA gain $G$ changes according to \eqref{eq:gain}. Thus, for different values of OMA, the variance of $\sigma_{\thn}^2$ in \eqref{eq:sigmas} will also change. The simulation parameters are summarized in Table~\ref{tab:my_label}, and the parameters specific to different modulation formats are summarized in Table~\ref{tab:mod_param}. 

\begin{table}[!t]
    \centering
    \caption{Simulation Parameters}
    \begin{tabular}{l c c}
    \hline \hline
       Laser RIN  & $\mathrm{RIN}$    &   $-140$ [dB/Hz]\\ \hline
       Modulator ER & $\mathrm{\varepsilon_r}$ & 5 [dB]\\ \hline
       Fiber Length &   $L$ &   $2$ [km]\\ \hline
       Attenuation  &   $\alpha$ & $0.35$ [dB/km]\\ \hline
       PD Responsivity  &   $\mathfrak{R}$    &   0.5 [A/W]\\ \hline
       PD-TIA Thermal Noise  &   $\sqrt{N_0/2}$   &   18 [pA/$\sqrt{\mathrm{Hz}}$]\\
       \hline
       \hline
    \end{tabular}    
    \label{tab:my_label}
\end{table}

\begin{table}[!t]
    \centering
    \caption{Modulation Parameters}
    \begin{tabular}{l c c c}
    \hline \hline
       Parameter & PAM-4 & PAM-6 & PAM-8\\ \hline
       Electrical Bandwidth $B$ & $68$ [GHz] & $52$ [GHz] & $45$ [GHz]\\
       Constellation $\min\{\mathcal{X}\}$ & -3 & -5 & -7\\
       Constellation $\max\{\mathcal{X}\}$ & +3 & +5 & +7\\
       Constellation IM Bias $\beta$ & 5.775 & 9.625 & 13.475\\
       \hline
       \hline
    \end{tabular}    
    \label{tab:mod_param}
\end{table}

The calculation of the SER in \eqref{eq:SER_analytical} is performed using the two different decision thresholds discussed in the previous section: the optimal threshold $r_i^\star$ obtained from \eqref{eq:r_opt}, and the commonly used approximation \cite{agrawal2012fiber} $\hat{r}_i$ from \eqref{eq:map_approx}. 
Numerical evaluation is performed via Monte Carlo simulations implementing the MAP decision rule from either~\eqref{eq:map_argmax0} or via the approximated thresholds \eqref{eq:map_approx}. To observe the effect that signal shaping has on the SERs and optimal decision thresholds, both geometric and probabilistic shaping are considered.

\subsection{Uniform and Equally-spaced Constellations}

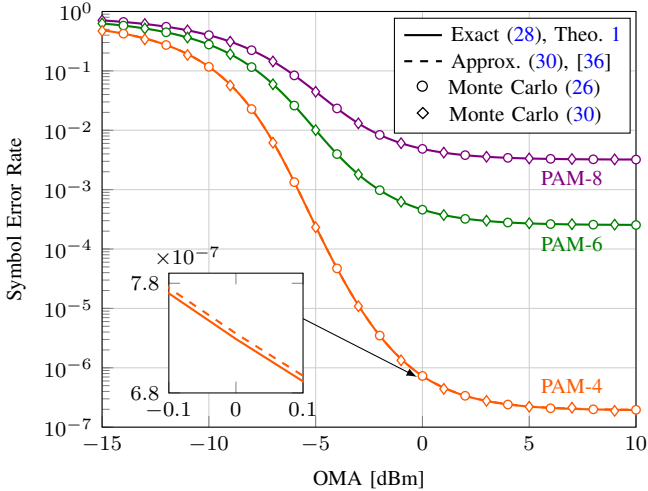
\begin{figure}[!t]
    \centering
    \input{FigTikz/SER_PAM_comparison}
    \caption{SER of uniform PAM-4, PAM-6, and PAM-8 over the signal-dependent noise channel from \eqref{eq:ch_eq1_full} with $\mathrm{RIN}=-140$ dB/Hz. The SER difference between the optimal MAP threshold and its approximation is negligible in this case.}
    \label{fig:ser_pam4}
\end{figure}

The results for equally-spaced constellations \mbox{$\mathcal{X}\!=\!\{\pm 1,\pm 3, \dotsc, \pm (M-1)\}$} with a uniform probability distribution $P_X(x_i)=1/M$, $\forall x_i\in\mathcal{X}$, are given in Fig.~\ref{fig:ser_pam4}. 
The analytical expression is plotted with solid and dashed lines. 
The solid line corresponds to using the optimal threshold $r_i^\star$ from \eqref{eq:r_opt}, whereas the dashed line is using the approximated threshold $\hat{r}_i$ from \eqref{eq:map_approx}. Furthermore, the circles correspond to Monte Carlo simulations using \eqref{eq:map_argmax0}, while  diamonds represent Monte Carlo simulations with the approximated thresholds \eqref{eq:map_approx}.
Figure~\ref{fig:ser_pam4} shows that the numerical simulations match exactly the analytical curves for the whole range of OMAs considered in the simulation. Note that the SER curves start to saturate after around 2 [dBm], which is an effect caused by the presence of RIN. The SNR for this type of channel has been studied in \cite[Eq.~(6)]{villenas2025ecoc}\footnote{The SNR calculated in \cite{villenas2025ecoc} is for a 2D constellation. However, the same expression can be extrapolated to the 1D constellations from this work, and the asymptotic behavior for large OMA still holds.}, where it is shown that as OMA increases, the SNR asymptotically converges to a fixed value \cite[Ch. 10.6]{hui2019introduction}. As a consequence, the number of symbol errors cannot be decreased further, which results in an error floor in the SER as can be observed in Fig.~\ref{fig:ser_pam4} for all PAM formats.

The inset in Fig.~\ref{fig:ser_pam4} shows that there is a negligible difference between the solid and dashed lines, and thus, for this uniform signaling scenario, the approximation in \eqref{eq:map_approx} results in a SER that is very close to the optimal one.
This conclusion applies to all PAM constellations considered in Fig.~\ref{fig:ser_pam4}.

\subsection{Shaped Constellations}

The results for shaped constellations are presented in Fig.~\ref{fig:ser_pam4_ps_opt}. Here we showcase the results for shaped PAM-6 with GS and PS. To present the results we use the same legend as in Fig.~\ref{fig:ser_pam4}. Furthermore, equally-spaced uniform PAM-6 is shown by the green curve as a reference from the previous plot.

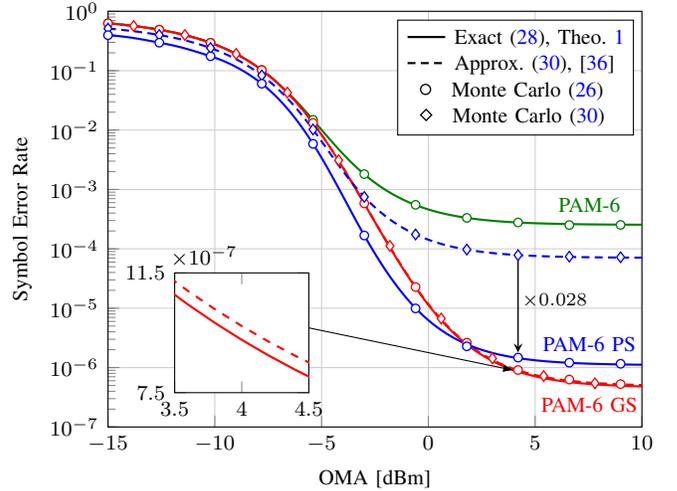
\begin{figure}[!t]
    \centering
    \input{FigTikz/PAM_GS_SER}
    \caption{SER of shaped PAM-6 with signal shaping over the signal-dependent noise channel. For PAM-6 with GS, the SER difference between using the two decision thresholds is negligible. For PAM-6 with PS, the use of the optimal decision threshold $r^\star$ achieves a significantly lower SER than that of using the approximation $\hat{r}$.}
    \label{fig:ser_pam4_ps_opt}
\end{figure}

\begin{figure}[!t]
    \centering
    \input{FigTikz/PAM_PSGS_SER_inset2}
    \vspace{-4mm}
    \caption{Optimal constellation $\mathcal{X}^\star$ found for GS (red) and input probability distribution $P_X(x)$ for PS (blue) at \mbox{$\mathrm{OMA}=0$} [dBm]. The top figure shows the uniform and equally-spaced case.}
    \label{fig:ser_pam4_ps_opt_inset}
\end{figure}
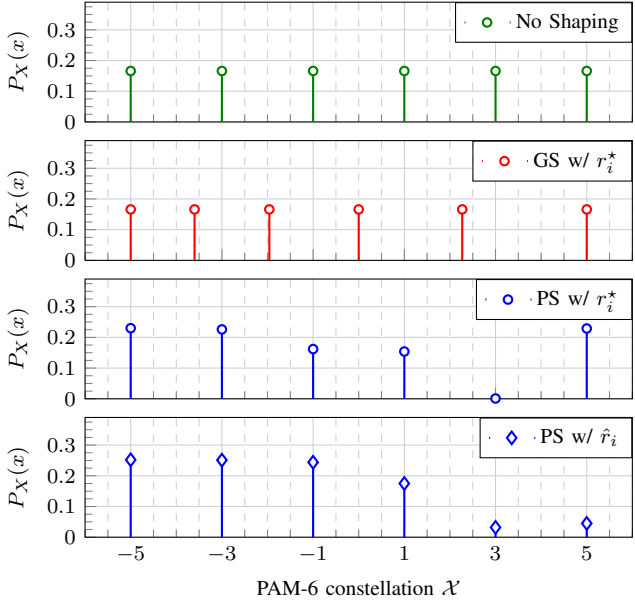

The red curve in Fig.~\ref{fig:ser_pam4_ps_opt} shows the SER results for \mbox{PAM-6} with GS. In this case, the input distribution $P_X(x)$ is still uniform, but the location of the constellation points is optimized. 
We optimize the locations $x_i$ for $i=0,1,\dotsc,M-1$ in order to minimize the SER at every OMA, i.e., we solve the optimization problem described as
\begin{equation} \label{eq:min_GS}
    \begin{split}
        \mathcal{X}^\star = \min_{\mathcal{X}}\quad\! &\mathrm{SER} \\
        \mathrm{s.t.}\quad\!  &x_0 = -(M-1) \quad \mbox{and} \quad x_{M-1} = +(M-1).
    \end{split}
\end{equation}
The points $x_0$ and $x_{M-1}$ are the outermost symbols of the constellation, and they are fixed at a constant value due to the peak-power constraint imposed by the IM \cite{che2023modulation,wiegart2020probabilistically}. 

The optimization problem from \eqref{eq:min_GS} is performed for the two different thresholds ($r_i^\star$ and $\hat{r}_i$), which will result in two optimal SER values for each case. We observe in Fig.~\ref{fig:ser_pam4_ps_opt} that by using GS (red curve), we can achieve a lower SER than uniform signaling (green curve) for the signal-dependent noise channel under consideration. This is consistent with the results from previous works where GS is employed to minimize the SER in presence of other types of signal-dependent noises. For example, in \cite{liang2023geometric} the signal-dependent noise source is the modulator chirp, whereas in \cite{zhou2022unequally} the main source of signal-dependent noise is the shot noise.

Moreover, the use of the approximate decision thresholds $\hat{r}_i$ (red diamonds) still achieves a SER very close to the optimal even with a nonequally-spaced constellation. The solid and dashed curves are almost overlapping in this case, as can be seen in the inset figure. Figure~\ref{fig:ser_pam4_ps_opt_inset} shows the uniform and equally-spaced constellation (green), and the optimized constellation obtained for $\mathrm{OMA}=0$ [dBm] (red). It can be seen that the innermost symbols from the GS optimized constellation $\mathcal{X}^\star$ are shifted slightly toward lower values with respect to that of the equally-spaced constellation. This results in the points having a lower noise variance in comparison, thus, increasing the noise tolerance of the shaped constellation.

Figure~\ref{fig:ser_pam4_ps_opt} also shows the SER results for PAM-6 with PS, i.e., when the input constellation is equally-spaced, but their probability distribution is optimized.
We optimize the distribution $P_X(x)$ in order to minimize the SER at every OMA, i.e., we solve the optimization problem described as 
\begin{equation} \label{eq:opt_PS}
    \begin{split}
       P_X(x)^\star = \min_{P_X(x)}\; & \mathrm{SER}\\
        \mathrm{s.t.}\quad\! & \sum_{x'\in\mathcal{X}} P_X(x') =1 \quad \mbox{and} \quad \mathbb{H}(X) \geq h_x,
    \end{split}
\end{equation}
where $\mathbb{H}(X)$ is the entropy of the constellation $X$. 
We set an entropy constraint to avoid that the optimal solution degenerates into a lower-order PAM format, i.e., zero or near-zero probability for any PAM-6 symbols. We choose a value of $h_x=2.30$, which is close to $\log_2(5)$ to avoid having more than one symbol with near zero probability. Similarly to the GS optimization, the optimization from \eqref{eq:opt_PS} is also performed for the two different thresholds ($\hat{r}_i$ and $r_i^\star$). The SER results for PS PAM-6 are shown with blue curves in Fig.~\ref{fig:ser_pam4_ps_opt}. 
Here, contrary to the uniform and GS cases, there is a significant difference between using the optimal and approximate decision thresholds.
Optimal threshold achieves almost two orders of magnitude lower SER than the approximate one, as shown by the black arrow. 
This is due to the fact that the approximate threshold \eqref{eq:map_approx} does not take into account the input distribution as opposed to the optimal expression from \eqref{eq:r_opt}. 
From Fig.~\ref{fig:ser_pam4_ps_opt_inset}, it can also be seen that the use of different thresholds also leads to different probability distributions. In the case of the approximate threshold, the two rightmost symbols have the lowest probability. In contrast, when using the optimal threshold the second to last symbol has the lowest probability, but the rightmost symbol has a much larger probability in comparison.

\begin{figure}[!t]
    \centering
    \input{FigTikz/PAM_PS_MI_single}
    \caption{Symbol error rate for PS PAM-8 for different $\mathbb{H}(X)$, calculated for the two different thresholds $r^\star$ and $\hat{r}$.}
    \label{fig:PAM_PS_MI}
\end{figure}
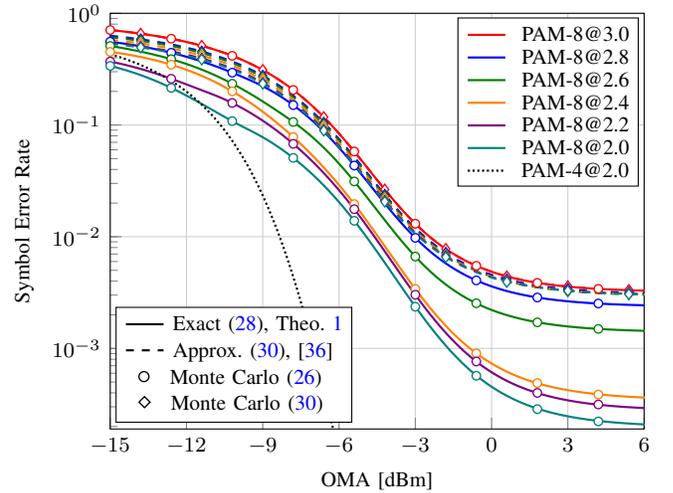

\subsection{Mutual Information}

In most of the works where PS is employed, the objective function or the optimization problem is the mutual information $I(X;Y)$ \cite{wiegart2020probabilistically}. Hence, here we analyze the case where the objective function is $I(X;Y)$ with channel transition probability $p_{Y|X}(y|x_i)$ in \eqref{eq:cond_pdf}. The optimization problem is now defined as
\begin{equation} \label{eq:MI_PS}
    \begin{split}
        P_X(x)^\star=\max_{P_X(x)}\; & I(X;Y)\\
        \mathrm{s.t.}\quad\! & \sum_{x'\in\mathcal{X}} P_X(x') =1.
    \end{split}
\end{equation}
For this new scenario, we consider a PAM-$8$ with PS, i.e., we optimize the input distribution $P_X(x)$ according to \eqref{eq:MI_PS}. 

The gains in mutual information for peak-power constrained IM-DD systems are known to be small (see, e.g., \cite{wiegart2020probabilistically} and \cite{zou2024evaluation}). However, here we aim to illustrate the accuracy of the SER calculation for a probabilistically-shaped IM-DD system, where the input distribution was optimized using a different cost function (the mutual information). From the optimized input distributions, we select 6 cases with different entropies $\mathbb{H}(X)$ ranging from $2.0$ to $3.0$, and we calculate SER for the two different thresholds expressions. Uniform and equally-spaced PAM-4 results are shown in Fig.~\ref{fig:PAM_PS_MI} as a reference (dotted black line). We denote the PS PAM-8 as PAM-$8$@$\mathbb{H}(X)$. The results in Fig.~\ref{fig:PAM_PS_MI} show again that there is a significant difference between the two thresholds. For the approximate thresholds $\hat{r}_i$, the SER curves (colored, dashed) are all grouped and result in values very close to the uniform distribution (PAM-$8$@$3.0$). Instead, for the optimal thresholds $r_i^\star$, the SER curves (solid) are distinguishable from each other and perfectly match Monte Carlo simulations. Therefore, these results show when employing PS, it is critical to use the optimized threshold $r^\star$ in order to obtain an accurate calculation of the SER. This is especially important in high data rate applications, where the OMA region of operation is high \cite{zhou2019beyond}.

%% file: FigTikz/SER_PAM_comparison.tex
\definecolor{my_green}{rgb}{0.55, 0.71, 0.0}    
\definecolor{my_purple}{rgb}{0.5, 0, 0.5}
\definecolor{azure}{rgb}{0, 0.5, 0.5}

\pgfmathsetmacro{\lw}{0.8}    
\pgfmathsetmacro{\lww}{0.5}    

\pgfplotstableread{FigTikz/data_txt/SER_PAM_comparison.txt}\datatable

\newcommand{\blacksolid}{\raisebox{2pt}{\tikz{\draw[color=black,solid,line width=\lw,mark=o,mark options={solid}](0,0) -- (5mm,0);}}}   
\newcommand{\blacksdashed}{\raisebox{2pt}{\tikz{\draw[color=black,dashed,line width=\lw,mark=o,mark options={solid}](0,0) -- (5mm,0);}}}   
\newcommand{\blacksdotted}{\raisebox{2pt}{\tikz{\draw[color=black,dotted,line width=0.75pt,mark=o,mark options={solid}](0,0) -- (5mm,0);}}}   
\newcommand{\blackdiamond}{\raisebox{0.5pt}{\tikz{\node[draw,line width=\lww,scale=0.4,diamond,color=black,fill=white](){};}}}
\newcommand{\blackcircle}{\raisebox{1pt}{\tikz{\node[draw,line width=\lww,scale=0.5,circle,color=black,fill=white](){};}}}
\newcommand{\ksquare}{\raisebox{1pt}{\tikz{\node[draw,line width=1pt,scale=0.6,rectangle,color=red,fill=white](){};}}}

\newcommand{\pamfour}{orange!70!red}

\begin{tikzpicture}
    \begin{semilogyaxis}[
    width=0.98\columnwidth,  
    height=2.8in,
    font=\footnotesize,
    xmin=-15, xmax=10,
    ymin=1e-7, ymax=1,
    xlabel={OMA [dBm]},
    ylabel={Symbol Error Rate},
    ytick pos=left,
    xtick pos=bottom,
    grid=major,
    grid style = {solid,lightgray!75},
    legend style = {legend pos=north east, font=\footnotesize, legend cell align=left, row sep=-0.5ex},
    ytickten={-7,-6,...,0},
    ]
    
        \pgfplotstableread{FigTikz/data_txt/SER_PAM_comparison_MC2.txt}\datatabletwo
        \pgfplotstableread{FigTikz/data_txt/SER_PAM_comparison_MC3.txt}\datatablethree

        \addplot[name path=SER_PAM8, color=my_purple, solid, line width=\lw] table[x index=0,y index=4] {\datatable};
        \addplot[name path=SER_PAM6, color=green!50!black, solid, line width=\lw] table[x index=0,y index=3] {\datatable};

        \addplot[name path=SER_PAM4, color=\pamfour, solid, line width=\lw] table[x index=0,y index=1] {\datatable};   
        \addplot[name path=SER_PAM4_alt, color=\pamfour, dashed, line width=\lw] table[x index=0,y index=2] {\datatable}; 

        \addplot[only marks, color=my_purple, solid, line width=\lww, mark=*, mark options={solid,scale=0.8,fill=white}] table[x index=0,y index=3]{\datatabletwo}; 
        \addplot[only marks, color=green!50!black, solid, line width=\lww, mark=*, mark options={solid,scale=0.8,fill=white}] table[x index=0,y index=2]{\datatabletwo};
        \addplot[only marks, color=\pamfour, solid, line width=\lww, mark=*, mark options={solid,scale=0.8,fill=white}] table[x index=0,y index=1]{\datatabletwo};

        \addplot[only marks, color=my_purple, solid, line width=\lww, mark=diamond*, mark options={solid,scale=1.0,fill=white}] table[x index=0,y index=3]{\datatablethree}; 
        \addplot[only marks, color=green!50!black, solid, line width=\lww, mark=diamond*, mark options={solid,scale=1.0,fill=white}] table[x index=0,y index=2]{\datatablethree};
        \addplot[only marks, color=\pamfour, solid, line width=\lww, mark=diamond*, mark options={solid,scale=1.0,fill=white}] table[x index=0,y index=1]{\datatablethree};

        \draw[draw,-latex] (axis cs:-6,8e-6) -- (axis cs:-0.25,7e-7);
        \node[font=\footnotesize] at (rel axis cs:0.16,0.41) {$\times10^{-7}$};

        \node[font=\footnotesize, \pamfour] at (axis cs:7, 5e-7) {PAM-4};
        \node[font=\footnotesize, green!50!black] at (axis cs:7, 1.2e-4) {PAM-6};
        \node[font=\footnotesize, my_purple] at (axis cs:7, 1.55e-3) {PAM-8};

        \node [draw,font=\footnotesize,fill=white,anchor=north east] at (rel axis cs:0.99, 0.985) {\shortstack[l]{
            \blacksolid \hspace{0.01cm} Exact \eqref{eq:r_opt}, Theo.~\ref{theo:optimal}\\
            \blacksdashed \hspace{0.01cm} Approx.~\eqref{eq:map_approx}, \cite{agrawal2012fiber}\\
            \hspace{0.08cm} \blackcircle \hspace{0.17cm} Monte Carlo~\eqref{eq:map_argmax0}\\
            \hspace{0.08cm} \blackdiamond \hspace{0.17cm} Monte Carlo \eqref{eq:map_approx}
        }};
        
    \end{semilogyaxis}

    \begin{semilogyaxis}[
        xshift=0.1\columnwidth,
        yshift=0.18in,
        width=0.38\columnwidth,
        height=1.25in,
        xmin=-0.1, xmax=0.1,
        ymin=6.8e-7, ymax=7.9e-7,
        xtick={-0.1,0,0.1},
        ytick={6.8e-7, 7.8e-7},
        yticklabels={$6.8$, $7.8$},
        axis background/.style={fill=white},
        font=\footnotesize,
    ]

        \addplot[color=\pamfour, solid, line width=\lw] table[x index=0,y index=1] {\datatable};   
        \addplot[color=\pamfour, dashed, line width=\lw] table[x index=0,y index=2] {\datatable};

    
    \end{semilogyaxis}


\end{tikzpicture}

%% file: FigTikz/PAM_GS_SER.tex
\definecolor{my_green}{rgb}{0.15, 0.5, 0.0}    
\definecolor{my_purple}{rgb}{0.5, 0, 0.5}

\pgfplotstableread{FigTikz/data_txt/SER_PAM_comparison_GS.txt}\datatableGS
\pgfplotstableread{FigTikz/data_txt/SER_PAM_comparison_GS_v2.txt}\datatableGSS
\pgfplotstableread{FigTikz/data_txt/SER_PAM_comparison_PS_Hx23.txt}\datatablePS
\pgfplotstableread{FigTikz/data_txt/SER_PAM_comparison_PSGS_MC_extra.txt}\datatableMC
\pgfplotstableread{FigTikz/data_txt/SER_PAM_comparison_PSGS_MC_extra2.txt}\datatableMCC

\pgfmathsetmacro{\lw}{0.8}    
\pgfmathsetmacro{\lww}{0.5}    

\begin{tikzpicture}
    \begin{semilogyaxis}[
    width=0.98\columnwidth,  
    height=2.8in,
    font=\footnotesize,
    xmin=-15, xmax=10,
    ymin=1e-7, ymax=1,
    xlabel={OMA [dBm]},
    ylabel={Symbol Error Rate},
    ytick pos=left,
    xtick pos=bottom,
    grid=major,
    grid style = {solid,lightgray!75},
    legend style = {legend pos=south west, font=\footnotesize, legend cell align=left, row sep=-0.5ex},
    ytickten={-7,-6,...,0},
    ]
        \addplot[name path=SER_PAM6, color=green!50!black, solid, line width=\lw] table[x index=0,y index=1] {\datatableGS};
        \addplot[name path=SER_PAM6, color=red, densely dashed, line width=\lw] table[x index=0,y index=2] {\datatableGSS};
        \addplot[name path=SER_PAM6, color=red, solid, line width=\lw] table[x index=0,y index=1] {\datatableGSS};
        \addplot[name path=SER_PAM6, color=blue, densely dashed, line width=\lw] table[x index=0,y index=3] {\datatablePS};
        \addplot[name path=SER_PAM6, color=blue, solid, line width=\lw] table[x index=0,y index=2] {\datatablePS};


        
        \addplot[only marks, color=green!50!black, solid, line width=\lww, mark=*, mark options={solid,scale=0.8,fill=white}] table[x index=0,y index=3]{\datatableMC};
        \addplot[only marks, color=red, solid, line width=\lww, mark=*, mark options={solid,scale=0.8,fill=white}] table[x index=0,y index=4]{\datatableMC};
        \addplot[only marks, color=red, solid, line width=\lww, mark=diamond*, mark options={solid,scale=1.0,fill=white}] table[x index=0,y index=1]{\datatableMCC};
        \addplot[only marks, color=blue, solid, line width=\lww, mark=diamond*, mark options={solid,scale=1.0,fill=white}] table[x index=0,y index=2]{\datatableMC};
        \addplot[only marks, color=blue, solid, line width=\lww, mark=*, mark options={solid,scale=0.8,fill=white}] table[x index=0,y index=1]{\datatableMC};


    
        \newcommand{\blacksolid}{\raisebox{2pt}{\tikz{\draw[color=black,solid,line width=\lw,mark=o,mark options={solid}](0,0) -- (5mm,0);}}}   
        \newcommand{\blacksdashed}{\raisebox{2pt}{\tikz{\draw[color=black,densely dashed,line width=\lw,mark=o,mark options={solid}](0,0) -- (5mm,0);}}}   
        \newcommand{\blacksdotted}{\raisebox{2pt}{\tikz{\draw[color=black,dotted,line width=0.75pt,mark=o,mark options={solid}](0,0) -- (5mm,0);}}}   
        \newcommand{\blackdiamond}{\raisebox{0.5pt}{\tikz{\node[draw,line width=\lww,scale=0.4,diamond,color=black,fill=white](){};}}}
        \newcommand{\blackcircle}{\raisebox{0.5pt}{\tikz{\node[draw,line width=\lww,scale=0.5,circle,color=black,fill=white](){};}}}
        \newcommand{\ksquare}{\raisebox{1pt}{\tikz{\node[draw,line width=1pt,scale=0.6,rectangle,color=red,fill=white](){};}}}

        \node [draw,font=\footnotesize,fill=white,anchor= north east] at (rel axis cs:0.985,0.98) {\shortstack[l]{
            \blacksolid \hspace{0.01cm} Exact \eqref{eq:r_opt}, Theo. \ref{theo:optimal}\\
            \blacksdashed \hspace{0.01cm} Approx. \eqref{eq:map_approx}, \cite{agrawal2012fiber}\\
            \hspace{0.08cm} \blackcircle \hspace{0.17cm} Monte Carlo~\eqref{eq:map_argmax0}\\
            \hspace{0.08cm} \blackdiamond \hspace{0.17cm} Monte Carlo \eqref{eq:map_approx}
        }};

        \draw[draw,-latex] (axis cs:-6,5e-6) -- (axis cs:4,9e-7);
        \node[font=\footnotesize] at (rel axis cs:0.18,0.41) {$\times10^{-7}$};

        \draw[->, >=stealth, line width=\lww] (axis cs:4.2, 6.5e-5) -- (axis cs:4.2, 1.8e-6) node[midway, right,font=\scriptsize,yshift=2pt,xshift=-2pt] 
        {$\times0.028$};]
        
    \end{semilogyaxis}

    \begin{semilogyaxis}[
        xshift=0.1\columnwidth,
        yshift=0.18in,
        width=0.38\columnwidth,
        height=1.25in,
        xmin=3.5, xmax=4.5,
        ymin=7.5e-7, ymax=1.15e-6,
        xtick={3.5,4,4.5},
        ytick={7.5e-7, 1.15e-6},
        yticklabels={$7.5$, $11.5$},
        axis background/.style={fill=white},
        font=\footnotesize,
    ]

    \addplot[color=red, solid, line width=\lw] table[x index=0,y index=1] {\datatableGSS};   
    \addplot[color=red, dashed, line width=\lw] table[x index=0,y index=2] {\datatableGSS};


    \end{semilogyaxis}

    


    \node[color=green!50!black, font=\footnotesize] at (6.4,2.95) {\shortstack[l]{ PAM-6 }};
    \node[color=blue, font=\footnotesize] at (6.4,1.1) {\shortstack[l]{ PAM-6 PS }};
    \node[color=red, font=\footnotesize] at (6.4,0.3) {\shortstack[l]{ PAM-6 GS }};

\end{tikzpicture}

%% file: FigTikz/PAM_PSGS_SER_inset2.tex
\definecolor{my_green}{rgb}{0.15, 0.5, 0.0}    
\definecolor{my_purple}{rgb}{0.5, 0, 0.5}

\pgfmathsetmacro{\lw}{0.8}    
\pgfmathsetmacro{\lww}{0.5}    

\begin{subfigure}[b]{\columnwidth}
\begin{tikzpicture}
    \begin{axis}[
        width=1\columnwidth,
        height=1.25in,
        font=\footnotesize,
        xtick pos=bottom,
        ytick pos=left,
        ymin=0, ymax=0.39,
        xmin=-6, xmax=6,
        minor tick num=3,
        grid=both,
        major grid style = {solid,lightgray!75},
        minor grid style = {dashed,lightgray!75},
        yminorgrids=false,
        ylabel={$P_X(x)$},
        ytick={0,0.1,...,0.4},
        xtick={-7,-5,...,7},
        xticklabels=\empty,
        legend style = {at={(1,1)}, anchor=north east, font=\footnotesize, legend cell align=left, row sep=-0.5ex},
        ]
       
        \addplot[ycomb, color=green!50!black, mark=*, line width = 0.8pt, mark options={solid,scale=0.8,fill=white}] 
        coordinates{
            (5, 0.166)
            (3, 0.166)
            (1, 0.166)
            (-1, 0.166)
            (-3, 0.166)
            (-5, 0.166)
            };
        \addlegendentry{No Shaping}
    \end{axis}
\end{tikzpicture}
\end{subfigure}

\begin{subfigure}[b]{\columnwidth}
\begin{tikzpicture}
    \begin{axis}[
        width=1\columnwidth,
        height=1.25in,
        font=\footnotesize,
        xtick pos=bottom,
        ytick pos=left,
        ymin=0, ymax=0.39,
        xmin=-6, xmax=6,
        minor tick num=3,
        grid=both,
        major grid style = {solid,lightgray!75},
        minor grid style = {dashed,lightgray!75},
        yminorgrids=false,
        ylabel={$P_X(x)$},
        ytick={0,0.1,...,0.4},
        xtick={-7,-5,...,7},
        xticklabels=\empty,
        legend style = {at={(1,1)}, anchor=north east, font=\footnotesize, legend cell align=left, row sep=-0.5ex},
        ]
       
        \addplot[ycomb, color=red, mark=*, line width = 0.8pt, mark options={solid,scale=0.8,fill=white}] 
        coordinates{
            (5, 0.166)
            (2.27, 0.166)
            (0, 0.166)
            (-1.96, 0.166)
            (-3.6, 0.166)
            (-5, 0.166)
            };
        \addlegendentry{GS w/ $r_i^\star$}
    \end{axis}
\end{tikzpicture}
\end{subfigure}

\begin{subfigure}[b]{\columnwidth}
\begin{tikzpicture}
    \begin{axis}[
        width=1\columnwidth,
        height=1.25in,
        font=\footnotesize,
        xtick pos=bottom,
        ytick pos=left,
        ymin=0, ymax=0.39,
        xmin=-6, xmax=6,
        minor tick num=3,
        grid=both,
        major grid style = {solid,lightgray!75},
        minor grid style = {dashed,lightgray!75},
        yminorgrids=false,
        ylabel={$P_X(x)$},
        ytick={0,0.1,...,0.4},
        xtick={-7,-5,...,7},
        xticklabels=\empty,
        legend style = {at={(1,1)}, anchor=north east, font=\footnotesize, legend cell align=left, row sep=-0.5ex},
        ]
                          
        \addplot[ycomb, color=blue, mark=*, line width = 0.8pt, mark options={solid,scale=0.8,fill=white}] 
        coordinates{
            (5, 0.229)
            (3, 0.001)
            (1, 0.154)
            (-1, 0.162)
            (-3, 0.226)
            (-5, 0.230)
            };
        \addlegendentry{PS w/ $r_i^\star$}
    \end{axis}
\end{tikzpicture}
\end{subfigure}

\begin{subfigure}[b]{\columnwidth}
\begin{tikzpicture}
    \begin{axis}[
        width=1\columnwidth,
        height=1.25in,
        font=\footnotesize,
        xtick pos=bottom,
        ytick pos=left,
        ymin=0, ymax=0.39,
        xmin=-6, xmax=6,
        minor tick num=3,
        grid=both,
        major grid style = {solid,lightgray!75},
        minor grid style = {dashed,lightgray!75},
        yminorgrids=false,
        ylabel={$P_X(x)$},
        xlabel={PAM-6 constellation $\mathcal{X}$},
        ytick={0,0.1,...,0.4},
        xtick={-7,-5,...,7},
        legend style = {at={(1,1)}, anchor=north east, font=\footnotesize, legend cell align=left, row sep=-0.5ex},
        ]
                          
        \addplot[ycomb, color=blue, mark=diamond*, line width = 0.8pt, mark options={solid,scale=1.2,fill=white}] 
        coordinates{
            (5, 0.045)
            (3, 0.032)
            (1, 0.175)
            (-1, 0.244)
            (-3, 0.251)
            (-5, 0.252)
            };
        \addlegendentry{PS w/ $\hat{r}_i$}
    \end{axis}
\end{tikzpicture}
\end{subfigure}

%% file: FigTikz/PAM_PS_MI_single.tex
\definecolor{my_green}{rgb}{0.55, 0.71, 0.0}    
\definecolor{my_purple}{rgb}{0.5, 0, 0.5}
\definecolor{azure}{rgb}{0, 0.5, 0.5}

\pgfmathsetmacro{\lw}{0.8}    
\pgfmathsetmacro{\lww}{0.5}    

\newcommand{\blacksolid}{\raisebox{2pt}{\tikz{\draw[color=black,solid,line width=\lw,mark=o,mark options={solid}](0,0) -- (5mm,0);}}}   
\newcommand{\blacksdashed}{\raisebox{2pt}{\tikz{\draw[color=black,dashed,line width=\lw,mark=o,mark options={solid}](0,0) -- (5mm,0);}}}   
\newcommand{\blackdotted}{\raisebox{2pt}{\tikz{\draw[color=black,densely dotted,line width=\lw,mark=o,mark options={solid}](0,0) -- (5mm,0);}}}   

\newcommand{\blackdiamond}{\raisebox{0.5pt}{\tikz{\node[draw,line width=\lww,scale=0.4,diamond,color=black,fill=white](){};}}}
\newcommand{\blackcircle}{\raisebox{0.5pt}{\tikz{\node[draw,line width=\lww,scale=0.5,circle,color=black,fill=white](){};}}}

\newcommand{\colorsolid}[1]{\raisebox{2pt}{\tikz{\draw[color=#1,solid,line width=\lw,mark=o,mark options={solid}](0,0) -- (5mm,0);}}}

\pgfplotsset{colormap={bluered}{rgb=(1,0,0) rgb=(0,1,0) rgb=(0,0,1)}}

\pgfplotsset{colormap={custommap}{rgb=(1,0,0);rgb=(0,0,1);rgb=(0,0.5,0);rgb=(1,0.5,0);rgb=(0.5,0,0.5);rgb=(0,0.5,0.5)}}

\pgfplotscreateplotcyclelist{mylist}{
    {red},{blue},{green!50!black},{orange},{blue!50!red},{green!50!blue},{black,densely dotted}
}

\begin{tikzpicture}
    \begin{semilogyaxis}[
    width=0.98\columnwidth,  
    height=2.8in,
    font=\footnotesize,
    xmin=-15, xmax=6,
    ymin=1.9e-4, ymax=1,
    xlabel={OMA [dBm]},
    ylabel={Symbol Error Rate},
    ytick pos=left,
    xtick pos=bottom,
    grid=major,
    grid style = {solid,lightgray!75},
    legend style = {at={(0.985,0.99)}, anchor=north east, font=\footnotesize, legend cell align=left, row sep=-0.5ex},
    xtick={-15,-12,...,10},
    ytickten={-4,-3,...,0},
    cycle list name=mylist,
    ]
    
    \pgfplotstableread{FigTikz/data_txt/PAM8_PS_200G_RIN-140_2.txt}\datatable
    \pgfplotstableread{FigTikz/data_txt/SER_PAM_comparison_PS_MI_MC.txt}\datatableMC
    \pgfplotstableread{FigTikz/data_txt/SER_PAM_comparison_PS_MI_MC_extra.txt}\datatableMCC
    
    \foreach \col in {13,...,8} {
        %

        
        \addplot+[line width=\lw] table[x index=0, y index=\col] {\datatable};
    }
    \addplot+[line width=\lw] table[x index=0, y index=14] {\datatable};
    
    \foreach \col in {20,...,15} {
        %

        
        \addplot+[line width=\lw, dashed] table[x index=0, y index=\col] {\datatable};
    }

    \pgfplotsset{cycle list shift=0}
    \addplot+ coordinates{(nan,nan)};
    \foreach \col in {6,5,...,1}{
        \addplot+[only marks, solid, line width=\lww, mark=diamond*, mark options={solid,scale=1,fill=white}] table[x index=0, y index=\col] {\datatableMCC};
    }
    
    \addplot+ coordinates{(nan,nan)};
    \foreach \col in {11,9,7,5,3,1}{
        \addplot+[only marks, solid, line width=\lww, mark=*, mark options={solid,scale=0.8,fill=white}] table[x index=0, y index=\col] {\datatableMC};
    }

    \node [draw,font=\footnotesize,fill=white,anchor= south west] at (rel axis cs:0.01,0.015) {\shortstack[l]{
        \blacksolid \hspace{0.01cm} Exact \eqref{eq:r_opt}, Theo. \ref{theo:optimal}\\
        \blacksdashed \hspace{0.01cm} Approx. \eqref{eq:map_approx}, \cite{agrawal2012fiber}\\
        \hspace{0.08cm} \blackcircle \hspace{0.17cm} Monte Carlo~\eqref{eq:map_argmax0}\\
        \hspace{0.08cm} \blackdiamond \hspace{0.17cm} Monte Carlo \eqref{eq:map_approx}
    }};

    \node [draw,font=\footnotesize,fill=white,anchor= north east] at (rel axis cs:0.99,0.99) {\shortstack[l]{
        \colorsolid{red} \hspace{0.01cm} PAM-8@3.0\\
        \colorsolid{blue} \hspace{0.01cm} PAM-8@2.8\\
        \colorsolid{green!50!black} \hspace{0.01cm} PAM-8@2.6\\
        \colorsolid{orange} \hspace{0.01cm} PAM-8@2.4\\
        \colorsolid{blue!50!red} \hspace{0.01cm} PAM-8@2.2\\
        \colorsolid{green!50!blue} \hspace{0.01cm} PAM-8@2.0\\
        \blackdotted \hspace{0.1cm} PAM-4@2.0
    }};
    
    \end{semilogyaxis}

\end{tikzpicture}

%% file: sections/5.conclusions.tex
\section{Conclusions}
\label{sec:conclusions}

We have provided a general analytical expression to calculate the SER for an IM-DD system that is subject to signal-dependent laser intensity noise. The developed expression is novel and makes no assumptions about the locations or distribution of the constellation points, making it compatible with GS and PS. Monte Carlo simulations validated the expression for different shaping scenarios and modulation formats. When comparing the presented SER expression with the one using an approximate threshold from the literature, we observed that using our proposed expression is critical when PS is employed. The usage of our expression results in an accurate calculation of the SER under PS, whereas the approximated thresholds resulted in an overestimation of the SER. On the other hand, the difference with the approximation was negligible for the case when GS is employed, and for the case when the constellation is equally-spaced and equiprobable. 
Future work includes extending the analysis in this paper to information rate maximization with probabilistic shaping. In particular, to coded systems using hard-decision FEC codes, which depend explicitly on the decision thresholds.

%% file: sections/A1.appendix.tex
\appendices
\section{Derivation of \eqref{eq:Y_k}} \label{app:A}
From \eqref{eq:Y_k_prev} we have that the two terms of interest to expand are $\tilde{P}_{\mathrm{tx}}(kT)$ and $\mathbb{E}[\tilde{P}_{\mathrm{tx}}(kT)]$. For the first one we have
\begin{align}
    \tilde{P}_{\mathrm{tx}}(kT) &\overset{\eqref{eq:P_tx1}}{=} \tilde{P}_{\mathrm{cw}}(kT)f(u(kT))\\ 
    &\overset{\eqref{eq:U_k}}{=} \tilde{P}_{\mathrm{cw}}(kT) f\left(\sum_{\ell=0}^\infty U_\ell p(kT-\ell T)\right)\\ \label{eq:f_uk}
    &=\tilde{P}_{\mathrm{cw}}(kT)f(U_k)\\ \label{eq:f_gX}
    &=\tilde{P}_{\mathrm{cw}}(kT)f(g(X_k))\\
    &\overset{\eqref{eq:u_k}}{=}\tilde{P}_{\mathrm{cw}}(kT)f\left(f^{-1}\left(\frac{\eta}{\Pcw}(X_k+\beta)\right)\right)\\ \label{eq:P_tilde}
    &=\tilde{P}_{\mathrm{cw}}(kT)\frac{\eta}{\Pcw}(X_k+\beta)\\ \label{eq:P_tilde2}
    &=\Pcw(1+\tilde{n}_\rin(kT))\left(\frac{\eta}{\Pcw}(X_k+\beta)\right)\\ \label{eq:Ptx_kT}
    &= \eta(1+\tilde{n}_\rin(kT))(X_k+\beta),
\end{align}
where to pass from \eqref{eq:f_uk} to \eqref{eq:f_gX}, we use the definition \mbox{$U_k=g(X_k)$} from Sec.~\ref{sec:transmitter}, and to pass from \eqref{eq:P_tilde} to \eqref{eq:P_tilde2}, we use the filtered version of \eqref{eq:Pcw}. Also, \eqref{eq:f_uk} follows from the fact that $p(kT-\ell T)\neq 0$ iff $\ell=k$ for raised cosine pulses. 
Finally, using \eqref{eq:Ptx_kT} in \eqref{eq:Y_k}, we obtain \eqref{eq:Y_k_prev}.


\section{Proof of Theorem \ref{theo:optimal}} \label{app:B}
Using \eqref{eq:cond_pdf} in \eqref{eq:MAP_eq0} results in
\begin{equation} \label{eq:proof_0}
    \frac{P_X(x_i)}{\sigma_i}\mathrm{exp}\left(-\frac{(r_i-x_i)^2}{2\sigma_i^2}\right) = \frac{P_X(x_j)}{\sigma_j}\mathrm{exp}\left(-\frac{(r_i-x_j)^2}{2\sigma_j^2}\right).
\end{equation}
Using $\log(\cdot)$ in \eqref{eq:proof_0} the expression can be rearranged into
\begin{equation} \label{eq:r_quad_eq}
    \frac{(r_i-x_i)^2}{\sigma_i^2} = \frac{(r_i-x_j)^2}{\sigma_j^2} + 2\log\left(\frac{P_X(x_i)}{P_X(x_j)}\frac{\sigma_j}{\sigma_i}\right).
\end{equation}
Equation \eqref{eq:r_quad_eq} is a quadratic equation in terms of $r_i$, namely, $a r_i^2 +2b r_i + c - d = 0$, where
\begin{equation} \label{eq:r_abcd}
    \begin{split}
        \underbrace{(\sigma_j^2-\sigma_i^2)}_{=a}r_i^2 + 2&\underbrace{(\sigma_i^2 x_j - \sigma_j^2 x_i)}_{=b}r_i + \underbrace{\sigma_j^2 x_i^2 - \sigma_i^2 x_j^2}_{=c}\\
        & -\underbrace{2\sigma_i^2\sigma_j^2\log\left(\frac{P_X(x_i)}{P_X(x_j)}\frac{\sigma_j}{\sigma_i}\right)}_{=d} = 0.
    \end{split}
\end{equation}
The proof is completed by using $r_i^\star = \frac{-b + \sqrt{b^2 - a(c - d)}}{a}$ as the solution of \eqref{eq:r_abcd}.

%% file: main_ieee.bbl
\begin{thebibliography}{10}
\providecommand{\url}[1]{#1}
\csname url@samestyle\endcsname
\providecommand{\newblock}{\relax}
\providecommand{\bibinfo}[2]{#2}
\providecommand{\BIBentrySTDinterwordspacing}{\spaceskip=0pt\relax}
\providecommand{\BIBentryALTinterwordstretchfactor}{4}
\providecommand{\BIBentryALTinterwordspacing}{\spaceskip=\fontdimen2\font plus
\BIBentryALTinterwordstretchfactor\fontdimen3\font minus \fontdimen4\font\relax}
\providecommand{\BIBforeignlanguage}[2]{{%
\expandafter\ifx\csname l@#1\endcsname\relax
\typeout{** WARNING: IEEEtran.bst: No hyphenation pattern has been}%
\typeout{** loaded for the language `#1'. Using the pattern for}%
\typeout{** the default language instead.}%
\else
\language=\csname l@#1\endcsname
\fi
#2}}
\providecommand{\BIBdecl}{\relax}
\BIBdecl

\bibitem{zhou2019beyond}
X.~Zhou, R.~Urata, and H.~Liu, ``Beyond 1 {Tb/s} intra-data center interconnect technology: {IM-DD} {OR} coherent?'' \emph{Journal of Lightwave Technology}, vol.~38, no.~2, pp. 475--484, Jan. 2020.

\bibitem{che2023modulation}
D.~Che and X.~Chen, ``Modulation format and digital signal processing for {IM-DD} optics at post-200{G} era,'' \emph{Journal of Lightwave Technology}, Jan. 2024.

\bibitem{zhong2018digital}
K.~Zhong, X.~Zhou, J.~Huo, C.~Yu, C.~Lu, and A.~P.~T. Lau, ``Digital signal processing for short-reach optical communications: {A} review of current technologies and future trends,'' \emph{Journal of Lightwave Technology}, vol.~36, no.~2, pp. 377--400, Jan. 2018.

\bibitem{800G_MSA_}
\BIBentryALTinterwordspacing
800{G} {Pluggable} {MSA}. [Online]. Available: \url{https://www.800gmsa.com/}
\BIBentrySTDinterwordspacing

\bibitem{berikaa2022net}
E.~Berikaa, M.~S. Alam, and D.~V. Plant, ``Net 400-{Gbps}/$\lambda$ {IMDD} transmission using a single-{DAC} {DSP}-free transmitter and a thin-film lithium niobate {MZM},'' \emph{Optics Letters}, vol.~47, no.~23, pp. 6273--6276, Dec. 2022.

\bibitem{hossain2021single}
M.~S.-B. Hossain, J.~Wei, F.~Pittal{\`a}, N.~Stojanovi{\'c}, S.~Calabr{\`o}, T.~Rahman, T.~Wettlin, C.~Xie, M.~Kuschnerov, and S.~Pachnicke, ``Single-{Lane} 402 {Gb/s} {PAM-8} {IM/DD} transmission based on a single {DAC} and an {O-Band} commercial {EML},'' in \emph{Optoelectronics and Communications Conference (OECC)}, Hong Kong, China, Jul. 2021.

\bibitem{szczerba20124}
K.~Szczerba, P.~Westbergh, J.~Karout, J.~S. Gustavsson, {\AA}.~Haglund, M.~Karlsson, P.~A. Andrekson, E.~Agrell, and A.~Larsson, ``4-{PAM} for high-speed short-range optical communications,'' \emph{Journal of optical communications and networking}, vol.~4, no.~11, pp. 885--894, Oct. 2012.

\bibitem{safari2015efficient}
M.~Safari, ``Efficient optical wireless communication in the presence of signal-dependent noise,'' in \emph{IEEE International Conference on Communication Workshop (ICCW)}, London, UK, Jun. 2015.

\bibitem{van2018optimization}
R.~van~der Linden, N.-C. Tran, E.~Tangdiongga, and T.~Koonen, ``Optimization of flexible non-uniform multilevel {PAM} for maximizing the aggregated capacity in {PON} deployments,'' \emph{Journal of Lightwave Technology}, vol.~36, no.~12, pp. 2328--2336, Jun. 2018.

\bibitem{baveja201756}
P.~P. Baveja, M.~Li, D.~Wang, C.~Hsieh, H.~Zhang, N.~Ma, Y.~Wang, J.~Lii, Y.~Liang, C.~Wang \emph{et~al.}, ``56 {Gb/s} {PAM-4} directly modulated laser for {200G/400G} data-center optical links,'' in \emph{Optical Fiber Communications Conference (OFC)}, Los Angeles, USA, Mar. 2017.

\bibitem{villenas2025ecoc}
F.~Villenas, K.~Wu, Y.~C. G{\"u}ltekin, J.~Riani, and A.~Alvarado, ``A new 5-bit/{2D}-symbol modulation format for relative intensity noise-dominated {IM-DD} systems,'' \emph{Preprint}, 2025 (under review available online at https://arxiv.org/abs/2506.01761).

\bibitem{wettlin2020dsp}
T.~Wettlin, S.~Calabr{\`o}, T.~Rahman, J.~Wei, N.~Stojanovic, and S.~Pachnicke, ``{DSP for high-speed short-reach IM/DD systems using PAM},'' \emph{Journal of Lightwave Technology}, vol.~38, no.~24, pp. 6771--6778, Dec. 2020.

\bibitem{bocherer2015bandwidth}
G.~B{\"o}cherer, F.~Steiner, and P.~Schulte, ``Bandwidth efficient and rate-matched low-density parity-check coded modulation,'' \emph{IEEE Transactions on Communications}, vol.~63, no.~12, pp. 4651--4665, Dec. 2015.

\bibitem{fehenberger2016probabilistic}
T.~Fehenberger, A.~Alvarado, G.~B{\"o}cherer, and N.~Hanik, ``On probabilistic shaping of quadrature amplitude modulation for the nonlinear fiber channel,'' \emph{Journal of Lightwave Technology}, vol.~34, no.~21, pp. 5063--5073, Nov. 2016.

\bibitem{qu2019probabilistic}
Z.~Qu and I.~B. Djordjevic, ``On the probabilistic shaping and geometric shaping in optical communication systems,'' \emph{IEEE Access}, vol.~7, pp. 21\,454--21\,464, Feb. 2019.

\bibitem{chen2023orthant}
B.~Chen, W.~Ling, Y.~Lei, Z.~Liang, and X.~Xue, ``Orthant-symmetric four-dimensional geometric shaping for fiber-optic channels via a nonlinear interference model,'' \emph{Optics Express}, vol.~31, no.~10, pp. 16\,985--17\,002, May 2023.

\bibitem{bocherer2019probabilistic}
G.~B{\"o}cherer, P.~Schulte, and F.~Steiner, ``Probabilistic shaping and forward error correction for fiber-optic communication systems,'' \emph{Journal of Lightwave Technology}, vol.~37, no.~2, pp. 230--244, Jan. 2019.

\bibitem{liga2022model}
G.~Liga, B.~Chen, and A.~Alvarado, ``{Model-aided geometrical shaping of dual-polarization 4D formats in the nonlinear fiber channel},'' in \emph{Optical Fiber Communications Conference (OFC)}, San Diego, USA, Mar. 2022.

\bibitem{ozaydin2024optimization}
B.~Ozaydin, X.~Chen, and D.~Che, ``{An optimization method for probabilistic constellation shaping in peak-power constraint systems in the presence of peak enhancement effects},'' in \emph{Optical Fiber Communications Conference (OFC)}, San Diego, USA, Mar. 2024.

\bibitem{ozaydin2025optimal}
------, ``On optimal probabilistically shaped constellations for unamplified optical interconnects,'' in \emph{Optical Fiber Communications Conference (OFC)}, San Francisco, USA, Apr. 2025.

\bibitem{liang2025probabilistic}
E.~M. Liang and J.~M. Kahn, ``Probabilistic shaping distributions for optical communications,'' \emph{Journal of Lightwave Technology}, Feb. 2025.

\bibitem{katz2018level}
G.~Katz and E.~Sonkin, ``Level optimization of {PAM-4} transmission with signal-dependent noise,'' \emph{IEEE Photonics Journal}, vol.~11, no.~1, pp. 1--6, Feb. 2019.

\bibitem{liang2023geometric}
E.~M. Liang and J.~M. Kahn, ``Geometric shaping for distortion-limited intensity modulation/direct detection data center links,'' \emph{IEEE Photonics Journal}, pp. 1--17, Dec. 2023.

\bibitem{villenas2025ofc}
F.~Villenas, K.~Wu, Y.~C. G{\"u}ltekin, J.~Riani, and A.~Alvarado, ``{On geometric shaping for 400 Gbps IM-DD links with laser intensity noise},'' in \emph{Optical Fiber Communications Conference (OFC)}, San Francisco, USA, Apr. 2025.

\bibitem{anttonen2011error}
A.~Anttonen, A.~Kotelba, and A.~M\"{a}mmel\"{a}, ``Error performance of {PAM} systems using energy detection with optimal and suboptimal decision thresholds,'' \emph{Physical Communication}, vol.~4, no.~2, pp. 111--122, June 2011.

\bibitem{chagnon2014experimental}
M.~Chagnon, M.~Osman, M.~Poulin, C.~Latrasse, J.-F. Gagn{\'e}, Y.~Painchaud, C.~Paquet, S.~Lessard, and D.~Plant, ``Experimental study of 112 {Gb/s} short reach transmission employing {PAM} formats and {SiP} intensity modulator at 1.3 $\mu$m,'' \emph{Optics express}, vol.~22, no.~17, pp. 21\,018--21\,036, Aug. 2014.

\bibitem{li2023application}
L.~Li, ``Application of general {MAP} detection theory for {PAM} signals to the analysis of {TDECQ} and {SSPRQ},'' \emph{Journal of Lightwave Technology}, vol.~41, no.~18, pp. 5942--5950, May 2023.

\bibitem{zhou2022unequally}
J.~Zhou, L.~Gan, Y.~Chen, Z.~Yang, Q.~Yang, M.~Tang, D.~Liu, and S.~Fu, ``Unequally spaced {PAM-4} signaling enabled sensitivity enhancement of a simplified coherent receiver applied in a {UDWDM-PON},'' \emph{Optics Express}, vol.~30, no.~20, pp. 35\,369--35\,380, Sep. 2022.

\bibitem{seimetz2009high}
M.~Seimetz, \emph{High-Order Modulation for Optical Fiber Transmission}.\hskip 1em plus 0.5em minus 0.4em\relax Springer, 2009.

\bibitem{hui2019introduction}
R.~Hui, \emph{{Introduction to Fiber-Optic Communications}}.\hskip 1em plus 0.5em minus 0.4em\relax Academic Press, 2019.

\bibitem{agrawal2013semiconductor}
G.~P. Agrawal and N.~K. Dutta, \emph{Semiconductor lasers}.\hskip 1em plus 0.5em minus 0.4em\relax Springer Science \& Business Media, 2013.

\bibitem{yang2023digital}
M.~Yang, A.~Yang, P.~Guo, Z.~Zhao, T.~Xu, and W.~Wan, ``{Digital Pre-Distortion for Mach--Zehnder Modulators in IMDD Optical Systems},'' in \emph{Opto-Electronics and Communications Conference (OECC)}, Shanghai, China, Jul. 2023.

\bibitem{yu2014trade}
H.~Yu, D.~Ying, M.~Pantouvaki, J.~Van~Campenhout, P.~Absil, Y.~Hao, J.~Yang, and X.~Jiang, ``Trade-off between optical modulation amplitude and modulation bandwidth of silicon micro-ring modulators,'' \emph{Optics express}, vol.~22, no.~12, pp. 15\,178--15\,189, Jun. 2014.

\bibitem{alam2021net}
M.~S. Alam, X.~Li, M.~Jacques, Z.~Xing, A.~Samani, E.~El-Fiky, P.-C. Koh, and D.~V. Plant, ``{Net 220 Gbps/$\lambda$ IM/DD Transmssion in O-band and C-band with Silicon Photonic Traveling-wave MZM},'' \emph{Journal of Lightwave Technology}, vol.~39, no.~13, pp. 4270--4278, Apr. 2021.

\bibitem{xu2011impact}
F.~Xu, M.-A. Khalighi, and S.~Bourennane, ``Impact of different noise sources on the performance of {PIN}-and {APD}-based {FSO} receivers,'' in \emph{Proceedings of the 11th International Conference on Telecommunications}, Graz, Austria, Jun. 2011.

\bibitem{agrawal2012fiber}
G.~P. Agrawal, \emph{{Fiber-Optic Communication Systems}}.\hskip 1em plus 0.5em minus 0.4em\relax John Wiley \& Sons, 2012.

\bibitem{proakis2008digital}
J.~G. Proakis and M.~Salehi, \emph{{Digital Communications}}.\hskip 1em plus 0.5em minus 0.4em\relax McGraw-hill, 2008.

\bibitem{wiegart2020probabilistically}
T.~Wiegart, F.~Da~Ros, M.~P. Yankov, F.~Steiner, S.~Gaiarin, and R.~D. Wesel, ``Probabilistically shaped 4-{PAM} for short-reach {IM/DD} links with a peak power constraint,'' \emph{Journal of Lightwave Technology}, vol.~39, no.~2, pp. 400--405, Jan. 2021.

\bibitem{zou2024evaluation}
D.~Zou, W.~Wang, Q.~Sui, F.~Li, and Y.~Cai, ``Evaluation of capacity-achieving distributions for memoryless {IM-DD} fiber-optic channels,'' \emph{Journal of Lightwave Technology}, Apr. 2025.

\end{thebibliography}
